\documentclass[oneside,english,11pt]{amsart}
\usepackage[foot]{amsaddr}
\usepackage[T1]{fontenc}
\usepackage[latin1]{inputenc}
\usepackage{latexsym}
\usepackage{amssymb}
\usepackage{amsmath}
\usepackage{amsfonts}
\usepackage{amsthm}
\usepackage{mathrsfs}
\usepackage{marvosym}
\usepackage{dsfont}
\usepackage[htt]{hyphenat}
\usepackage{hyperref}
\usepackage{graphicx}

\usepackage{fancyhdr}

\usepackage{color}
\definecolor{rltred}{rgb}{0.75,0,0}
\definecolor{rltgreen}{rgb}{0,0.5,0}
\definecolor{rltblue}{rgb}{0,0,0.75}

\usepackage{geometry}
\newcommand{\bitem}{\begin{itemize}}
\newcommand{\eitem}{\end{itemize}}
\newcommand{\benu}{\begin{enumerate}}
\newcommand{\eenu}{\end{enumerate}}
\newcommand{\beq}{\begin{eqnarray*}}
\newcommand{\eeq}{\end{eqnarray*}}

\newcommand{\R}{\mathbb{R}}
\newcommand{\RR}{\mathbb{R}}

\newcommand{\E}{\mathbb{E}}
\newcommand{\F}{\mathcal{F}}
\newcommand{\G}{\mathcal{G}}

\newcommand{\C}{\mathbb{C}}
\newcommand{\N}{\mathbb{N}}
\newcommand{\NN}{\mathbb{N}}

\newcommand{\T}{ \mathbb{T}}
\newcommand{\Prob}{\mathbb{P}}

\newcommand{\ind}[1]{\mathds{1}_{#1}}
\DeclareMathOperator*{\esssup}{ess\,sup}
\DeclareMathOperator*{\essinf}{ess\,inf}

\DeclareMathOperator*{\argmin}{arg\,min}

\newtheoremstyle{miEstilo}
{5pt}
{5pt}
{\slshape} 
{\parindent}
{\bfseries}
{:}
{\newline } 
{}
\theoremstyle{miEstilo}

\theoremstyle{definition}
 \newtheorem{theorem}{Theorem}[section]
 \newtheorem{lemma}[theorem]{Lemma}
 \newtheorem{proposition}[theorem]{Proposition}
 \newtheorem{corollary}[theorem]{Corollary}

\theoremstyle{definition}
\newtheorem{definition}[theorem]{Definition}
\newtheorem{example}[theorem]{Example}
\newtheorem{assumption}[theorem]{Assumption}

\theoremstyle{remark}
\newtheorem{remark}[theorem]{Remark}

\begin{document}
\title{ Conditional Analysis and a Principal-Agent problem }


\author{ Julio Backhoff $^1$ }
  \address{ $^1$ Vienna University of Technology,
Institute of Statistics and Mathematical Methods in Economics,
Wiedner Hauptstr. 8 / E105-7 MSTOCH
1040 Wien, Austria. E-mail: \texttt{julio.backhoff@tuwien.ac.at} 
The author gratefully acknowledges financial support from the European Research Council (ERC) under grant FA506041 and the Austrian Science Fund (FWF) under grant Y00782, 
as well as from the Berlin Mathematical School. } 

\author{Ulrich Horst $^2$ }
  \address{  $^2$ Institut f\"ur Mathematik, Humboldt-Universit\"at zu Berlin.
  E-mail: \texttt{horst@mathematik.hu-berlin.de} Financial support through the SFB 649 {\sl Economic Risk} is gratefully acknowledged.}

\maketitle

\begin{abstract}
We analyze conditional optimization problems arising in discrete time Principal-Agent problems of delegated portfolio optimization with linear contracts. Applying tools from {\it Conditional Analysis} we show that some results known in the literature for very specific instances of the problem carry over to translation invariant and time-consistent utility functions in very general probabilistic settings. {However, we find that optimal contracts must in general make use of derivatives for compensation.}
\end{abstract}

\begin{description}
\item [{\bf Keywords}] Principal-Agent problem, conditional analysis, portfolio delegation, variational utilities.
\item [{\bf MSC 2010}] 91G10,91G80,91A25,46S10,90C39.
\item [{\bf JEL}]G110,C610,C650,D860.
\end{description}

\renewcommand{\baselinestretch}{1.1}\normalsize

\section{Introduction}
\label{introduction}

In this article we analyze conditional optimization problems arising in the dynamic \textit{Principal-Agent (PA) Problem} of delegated portfolio management. In these models, which belong to the class of contracting problems under moral hazard, an investor (the Principal) outsources her portfolio selection to a manager (the Agent) whose investment decisions the investor cannot or does not want to monitor. 

Moral hazard problems have been first studied in \cite{Holmstrom,HolmstromMilgrom} in static environments and in \cite{SS87,Schaettler93} in dynamic ones. In recent years, such problems have received renewed attention in the economics and financial mathematics literature. 
The seminal contribution \cite{Sannikov} analyzed dynamic moral hazard problems in continuous time in which the output is a diffusion process with drift determined by the Agent's effort. The optimal contract, based on the Agent's continuation value as a state variable, was computed using sophisticated stochastic control and PDE methods. Using similar tools, \cite{BMRV} studied a PA model in which a risk-neutral Agent with limited liability must exert unobservable effort to reduce the likelihood of large but infrequent losses. In \cite{Williams} a Stochastic Maximum Principle was applied to dynamic PA models to find first order conditions for optimality. In the most general case the Stochastic Maximum Principle leads to the characterization of optimal contracts through a system of fully coupled Forward-Backward Stochastic Differential Equations for which no general existence theory exists. These equations can typically only be solved explicitly when the analysis is confined to models driven by Brownian motion, specific preferences - typically linear, expected exponential or power utility functions - and information is symmetric, i.e. both parties observe the driving Brownian motion. We refer to the monograph \cite{Cvitanic} for a systematic survey of the mathematical literature on dynamic PA models and to \cite{Cvitanicnodominado} for a recent model of portfolio delegation under incomplete information which leads to even more complex dynamics.  

Our work is motivated by that of Ou-Yang in \cite{Ou-Yang} - which was in a sense generalized in \cite{CvitanicJET} - where a delegated portfolio management problem in continuous time was analyzed. {In his model, the {Agent} observes prices (a geometric brownian motion) while the {Principal} observes prices and the fluctuations in wealth resulting from the Agent's investment strategy; investment decisions are unobservable to the Principal and known only to the Agent.} {Under the assumption of exponential utilities the contracting problem was solved by means of a HJB approach, finding that the optimal contract is of the form ``cash plus a convex combination of the generated wealth and a benchmark portfolio''; derivatives are not part of the optimal contract. 
%
%

{Our goal is to clarify the mathematical structure of the optimal contracting problem and to analyze under which conditions on the Principal's and Agent's preferences the main structure-of-equilibrium-contract results in \cite{Ou-Yang} carry over to more general probabilistic settings.} 
To this end, we consider a portfolio delegation model in discrete time, 
retain the assumptions on the contract space and information structure but allow for rather general utility functions and price dynamics.} Our main assumptions on preferences are time-consistency and translation invariance; such preferences have been extensively studied in the mathematical finance literature; see \cite{BarKa, OCE, Equilibrium, Divutil, FKV}. 
%
%
With our choice of preferences we prove that the problem of dynamic contract design can be reduced to a series of one-period conditional optimization problems of risk-sharing type under constraints and optimal contracts can be computed by backwards induction. To do so, we employ the usual approach of viewing the Agent's continuation utility at any point in time as the Principal's decision variable, with the Principal's decisions being restricted by an incentive compatibility constraint. To the best of our knowledge this argument was first put forward in \cite{SS87} and later in \cite{Schaettler93}.

Our approach of reducing the dynamic contracting problem to a series of conditional one-period problems is similar to the one employed in {\cite{Equilibrium} where a model of equilibrium pricing in incomplete markets was analyzed.} The optimization therein is simpler, though, as the exchange of risk takes place through linear subspaces spanned by the tradable assets which is not the case in our model. Our optimization problems can be viewed as conditional extensions of the ones analyzed in e.g.\ \cite{BarKa, Borch, Equilibrium} where the exchange of risk takes place through (conditional) $L^p$ spaces. Conditional analysis - see e.g.\ \cite{FKV,FKV2,CondRd} - provides a framework to tackle conditional optimization problems, at the same time avoiding technical measurable selection arguments. 

Our conditional one-period optimization problems are  not convex a-priori, due to incentive compatibility constraints. However, {with our choice of contracts} the Principal's and Agent's problems can be merged into unconstrained ones, {which if solvable} yield an optimal contract. In economic terms, the reduction to unconstrained problems means that the first-best solution is implementable under moral hazard if it exists: the contract that one obtains is the same that one would obtain if the Principal and Agent had the same information and had to share the gains and losses from trading between themselves so as to maximize aggregate utility. 

The intuition is that in computing an optimal contract the Principal computes the Agent's optimal actions as function of stock prices. {This resolves the asymmetry of information and leads eventually to our main result that the optimal contract - if it exists - is of the form ``cash plus a convex combination of the generated wealth and a benchmark portfolio plus a path-dependent derivative on the stock price process''. In particular, under an optimal contract the Principal fully surrenders to the Agent the wealth generated by trading in exchange for a benchmark portfolio plus a (generally non-replicable) derivative. 

As in \cite{Ou-Yang} the benchmark portfolio is related to the optimal (at-equilibrium) effort of the Agent. Unlike in \cite{Ou-Yang} derivatives are generally part of optimal compensation schemes. Derivatives are {\sl not} needed in Markovian models under a predictable representation property (PRP)\footnote{Loosely speaking the predictable representation property states that uncertainty is spanned by finitely many random factors.}. The latter includes a discrete-time version of the model in \cite{Ou-Yang} as well as many of the standard dynamic risk sharing problems under symmetric information as special cases. It is only in this (restricted) setting that we can prove in Proposition \ref{summmark} that the structure-of-contracts results in \cite{Ou-Yang} carry over to more general preferences as long as the Agent's and the Principal's preferences originate from a common base preference functional (e.g. exponential utilities).} 
%

The main challenge is then to solve the unconstrained optimization problems. The approach we follow is to prove that the set of potential optimizers is bounded in a suitable sense. In the greatest generality we work in the conditional version of $L^1$ spaces and with conditional utility functions enjoying a certain sequential upper-semicontinuity on balls, which in particular yields a variational representation of the preference functionals in the spirit of \cite{EpsteinWang, GilboaSchmeidler, MMR2006}.
%
%
%
 The transit from boundedness to optimality uses a form of the usual {Komlos argument}. With this we fully solve in Theorem \ref{casodina} the PA problem for {a class of}  Optimized Certainty Equivalent (OCE) utilities including Average Value at Risk, and bounded prices. 

In a Markovian framework under PRP our static conditional problems reduce to deterministic ones in Euclidean spaces. For such setting we find for general OCEs the optimal contract by the Lagrange multiplier method. Under PRP the solution to our contracting problem can also be obtained in terms of the solution to a coupled system of backward stochastic difference equations. As in \cite{Equilibrium}  the benefit of having a discrete model is that such systems can be solved by backwards induction, {whereas the continuous-time equivalent is usually intractable. This, and the fact that continuous-time models are unlikely to yield additional insights into the structure of optimal contracts over discrete-time models, motivates our discrete-time framework.}

The remainder of this paper is structured as follows. In Section \ref{themodel} we introduce the modeling framework including the preferences and the contract space. We show how the dynamic problem reduces to a sequence of static ones and present our main results along with examples (OCEs) for which these results can be applied. In Section \ref{GenResL0} we prove general attainability results for the Agent's and Principal's problem. In Section \ref{mark} we specialize our analysis assuming Markovianity and PRP, which allows us to obtain the optimal contracts explicitly. In Appendix A we survey existing and prove new conditional analysis results which we need throughout this work. Appendices B and C prove results on OCEs and one of the main results of this paper, respectively.  

{{\sl Notation.} We take the convention that vectors are regarded as column ones. The transpose of a vector $x$ is denoted $x'$ and unless necessary to do otherwise the inner product of two vectors $x,y$ is denoted $x y$. As usual, $(\cdot)_+$ and $(\cdot)_-$ denotes taking positive and negative parts.}

\section{The model and main results}
\label{themodel}

 We consider a discrete time model with time grid  $ \{0, 1, ..., T\}$ for some deterministic terminal time $T < \infty$. Uncertainty is modelled by a probability space $(\Omega, \F, \Prob)$. The probability space carries an $N$-dimensional, strictly positive, discounted stock price process $P=\{P_t\}$ whose filtration we denote by $\F=\{\F_t\}$. We assume throughout that $\E[P_{t+1}\vert\F_t]$ is finite\footnote{All equalities and inequalities are to be understood in the $\Prob$-a.s.\ sense.}. We put $\Delta P_{t+1} := P_{t+1}-P_t$ and $\Delta\tilde{P}_{t+1}:=diag(P_{t})^{-1}\Delta P_{t+1}$, where $diag(\cdot)$ denotes the diagonal matrix associated with the vector in its argument. The same notation applies for other processes different than $P$. We write {$P_{0:t}$ to denote the path of the price process from time $0$ to $t$.} For a $\sigma$-algebra ${\mathcal G}$ we denote by $L^0({\mathcal G})$ the set of real-valued ${\mathcal G}$-measurable functions. $\underline{L^0}({\mathcal G})$ and $\overline{L^0}({\mathcal G})$ denote the set of ${\mathcal G}$-measurable functions taking values in ${\mathbb R} \cup\{-\infty\}$, respectively ${\mathbb R} \cup\{+\infty\}$. 


\subsection{Effort levels and wealth dynamics}
\quad

At each time $t  \in \T := \{0, 1, ..., T-1\}$ the Agent (he) chooses an N-dimensional $\F_t$-measurable random variable  $A_t$ {that we call \textit{effort level}, in line with the Principal-Agent literature. For the delegated portfolio optimization application we have in mind the vector $A_t$ stands for the {dollar amount} invested in each asset. The cost associated with choosing $A_t$ is given by $c_t(A_t)$. We make the following standing assumption: }

\begin{assumption}\label{estconv}
	The cost functions $c_t(\cdot):\RR^N\to \RR$ are strictly convex for each $t  \in \T$. 
\end{assumption}

Effort levels are known only to the Agent. The wealth at time $t \in \T$ associated with a sequence of effort levels $A=\{A_t\}$ is given by\footnote{For simplicity we assume zero interest rate}:
\begin{equation}
W^A_t=W_0 +\Delta W^A_1+\dots + \Delta W^A_{t} = W_0 + \sum_{s<t}A_s \Delta\tilde{P}_{s+1} .\label{flow}
\end{equation} 
The Principal (she) observes progressively stock prices and wealth levels and offers the Agent a contract based on her available information. Following \cite{Ou-Yang} a contract will consist of a linear combination of a payment contingent on the evolution of the price process and a reward depending linearly on the wealth increments. This includes replicable derivatives on the terminal wealth. 

\subsection{Preferences} 
\label{preferences}
\quad

{
{
Payments are evaluated according to a family of time-consistent and translation invariant utility functions.  To connect with the existing literature we first define our preference functionals on spaces of almost surely finite random variables (``general framework''). Subsequently, we introduce an additional conditional integrability and continuity condition (``conditional $L^1$ framework'') from which we infer a variational representation of our preference functionals. While time-consistency and translation invariance allows us to reduce the dynamic contracting problem to a series of conditional one-step ones, the variational representation allows us to state sufficient conditions under which the Principal's and the Agent's conditional optimal one-step payments and actions exist at any time.}} 


\subsubsection{General framework}\label{genpreferences}

To introduce our preference functionals we denote by $\F^A$, for a given choice of effort level $A$, the filtration generated by the pair of processes $(P,W^A)$. For the Agent the filtrations $\F$ and $\F^A$ coincide; for the Principal they differ unless she knows the Agent's actions\footnote{The fact that the Principal observes price and wealth dynamics does not necessarily mean that she can observe directly Agent's decisions. For {\sl optimal} contracts the Principal will indeed know Agent's decisions as function of prices. This is not true ``off equilibrium'', though. Hence we need to distinguish Agent's and Principal's information at this point.}. The respective preferences are then encoded by a family of utility functionals:
\begin{equation}
U^a_t:L^0(\F_T)\to \underline{L}^0(\F_t) \hspace{6pt}\mbox { and }\hspace{6pt} U^p_t:L^0(\F^A_T)\to \underline{L}^0(\F^A_t) \quad (t \in \T). 
\end{equation}
%

 We use the notation $U^a$ and $U^p$ when referring to the Agent's and Principal's preferences. For a  filtration $\{\G_t\}$ and a family $U:=\{U_t\}$ of utility functionals $U_t:L^0(\G_T)\mapsto \underline{L}^0(\G_t)$ we say that $U$ is:
\begin{itemize}

\item \textit{normalized if } $U_t(0)=0$,

\item \textit{proper if }{there exists $X'\in L^0(\G_T) $ s.t.\ $U_t(X')>-\infty$ and $U_t(X)<\infty$ for all $X\in L^0(\G_T)$}

\item \textit{monotone if } $U_t(X)\geq U_t(Y)$ whenever $X,Y\in L^0(\G_T)$ and $X\geq Y$ 

\item $\F_t$-\textit{conditionally concave if } $U_t(\lambda X + (1-\lambda)Y)\geq \lambda U_t(X) + (1-\lambda)U_t(Y)$ whenever $\lambda \in L^0(\G_t) \cap [0,1]$ and $X,Y\in L^0(\G_T)$,

\item $\F_t$-\textit{translation invariant if } $U_t(X+Y)=U_t(X)+Y$ whenever $X\in L^0(\G_T)$ and $Y\in L^0(\G_t)$,

\item \textit{time consistent if } $U_{t+1}(X)\geq U_{t+1}(Y)$ implies $U_{t}(X)\geq U_{t}(Y) $, 

\end{itemize} 
for all $t \in \T$. We shall refer to these axioms as the \textit{usual conditions/assumptions} and denote by 
$$dom(U_t):=\{X\in L^0(\G_T):U_t(X)\in L^0(\G_t)\},$$ 
the domain of $U_t$. For a detailed discussion of the usual conditions along with their implications for utility optimization and equilibrium pricing we refer to \cite{Equilibrium} and references therein. For instance, it is well-known that they imply the \textit{tower property}, stating that $U_t(X)=U_t(U_{t+1}(X))$ whenever $X\in dom(U_{t+1})$, as well as the \textit{local property}, stating that $\ind{A}U_t(X)=\ind{A}U_t(Y)$ whenever $X,Y\in L^0(\G_T), A \in \G_t$ and $\ind{A}X=\ind{A}Y$.  

{We assume throughout that $U^a_t$ and $U^p_t$ satisfy the usual conditions w.r.t.\ the respective filtrations. They are satisfied for a wide class of preferences as illustrated by the following examples.} 

\begin{example}[\bf Entropic utilites]
\label{ejementropic}
Given a constant $\gamma>0$ the entropic family given by $$U_t(X)=-\frac{1}{\gamma}\log \E\left [ \exp\left ( -\gamma X\right ) |\G_t \right ],$$
evidently satisfies the usual conditions.
\end{example}  

\begin{example}[\bf Pasting]
\label{ejempaste}
{Starting from one-step utilities defined for bounded random variables, a family satisfying the usual conditions can be built over $L^0$ as follows \cite[Example 2]{Equilibrium}.} For each $t,$ let $\tilde{U}_t:L^{\infty}(\G_{t+1})\to L^{\infty}(\G_t)$ be a {normalized} and $\G_t-${translation invariant} functional, for which the extensions 
\begin{equation}
X\mapsto \lim_{n\rightarrow +\infty}\lim_{m\rightarrow {-}\infty} \tilde{U}_t([X\wedge n]\vee m) ,\label{eqextension}
\end{equation}
again denoted $\tilde{U}_t$, are well defined between $\underline{L}(\G_{t+1})$ and $\underline{L}(\G_t)$. It is not difficult to see that the \textit{pasting} $U_t(X):=\tilde{U}_t \circ \tilde{U}_{t+1}\circ \dots \circ \tilde{U}_T (X)$ forms a time consistent and translation invariant family.

\end{example}

\begin{example}[\bf Optimized Certainty Equivalents]
\label{oce}
{
Consider $H_t(\cdot)$ a convex, closed and increasing function satisfying $H_t^*(1):=\sup_s[s-H_t(s)]=0$, and define the one-step functionals
$$\tilde{U}_t:X\in L^{\infty}(\G_{t+1}) \mapsto \esssup_{s\in\RR}\{s-\E[H_t(s-X)\vert \G_t]\},$$ 
which are then normalized, translation-invariant and monotone. Such a family is called \textit{Optimized Certainty Equivalent} (OCE) in the literature; see \cite{OCE}. {The entropic utility of Example \ref{ejementropic} corresponds to $H(l)=\gamma^{-1}\exp(\gamma l-1)$. Lemma \ref{ejemx2} shows that if $1\in int(dom(H_t^*))$ and $H_t$ is bounded from below, the extensions \eqref{eqextension} are well-defined. Hence we obtain a family satisfying the usual conditions by pasting; see Remark \ref{remtechnical} as well. This fills a minor gap in \cite{Equilibrium}.} }
\end{example}

\begin{example}[\bf Tail-value-at-risk utility]
\label{ejemTVAR}
{{By Lemma \ref{ejemx2}},
a family satisfying all the requirements of Example \ref{oce} is given by the so-called Tail-value-at-risk (TVAR) utilities, defined for each $\lambda\in (0,1)$ by 
\begin{equation}
\tilde{U}_t(X)= \esssup_{s}\left\{s-\lambda^{-1}\E([s-X]_+|\G_t)  \right\}.\label{TVARrep}
\end{equation}
{
TVAR (or Average-Value-at-Risk) was characterized in \cite{Rock-Uryasev} and later extensively analyzed in the mathematical finance literature. The representation \eqref{TVARrep} is more convenient for us than the equivalent: 
$$ \tilde{U}_t(X) = -\frac{1}{\lambda}\int_{1-\lambda}^1 V@R_{\alpha}(-X |\G_t)d\alpha.$$}
}
\end{example}

\subsubsection{Conditional $L^1$ framework}\label{preferencesL1}
 
{We are now going to introduce additional conditional integrability and continuity conditions on our preference functionals (we refer to Appendix A for more details). We define} for two sigma-algebras $\mathcal{G}\subset\tilde{\mathcal{G}}$ the \textit{conditional $L^1$ space} $$L^1_{\mathcal{G}}(\tilde{\mathcal{G}}):=\left\{Z \in L^0(\tilde{\mathcal{G}}):\E[\vert Z\vert\vert\mathcal{G}]\in L^0(\mathcal{G})\right\}.$$ For $p<\infty$ the $L^p$ variant thereof is evident and we remark that $L^1_{\mathcal{G}}(\tilde{\mathcal{G}})= L^0(\mathcal{G})L^1(\tilde{\mathcal{G}})$ as sets. Call also 
\[
	L^{\infty}_{\mathcal{G}}(\tilde{\mathcal{G}}):=\{Z \in L^0(\tilde{\mathcal{G}}):\vert Z\vert \leq Y,\mbox{ for some }Y\in 		L^0(\mathcal{G})\}. 
\]	
{The following continuity property can be viewed as a Fatou property in our conditional framework.} 
%
%
%

\begin{definition}
\label{ranx}
For $p\in[1,\infty)$ a functional $U:L^p_{\mathcal{G}}(\tilde{\mathcal{G}})\rightarrow \underline{L}^0(\mathcal{G})$ is called $L^0-L^p$ upper semicontinuous if for each sequence $\{X_n\}_n$ bounded in $L^p_{\mathcal{G}}(\tilde{\mathcal{G}})$ (i.e.\ $\sup_n \E[|X_n|^p|\mathcal{G}]\in L^0(\mathcal{G})$) such that $X_n\rightarrow X$ a.s.\ it holds that $\limsup U(X_n)\leq U(X)$. {We use this terminology even if $U$ is defined in a larger set than $L^p_{\mathcal{G}}(\tilde{\mathcal{G}})$.}
\end{definition}

We state now a standing assumption on the preferences. 



\begin{assumption}
\label{suposUasymm}
Let $U$ stand for $U^a$ or $U^p$ and $\G$ for $\F$ or $\F^A$, respectively. Then $U$ satisfies the {usual conditions} with respect to $\G$. Moreover, $U_t$ is $L^0-L^1$ upper semicontinuous for each $t$ and $${\{X_{-}:X\in dom(U_t)\} }\subset L^1_{\G_t}(\G_T).$$ 
\end{assumption}

%
%


{The following representation result is an immediate consequence of Proposition \ref{implicanciausc}. It will be used below to prove that the Agent's one-step optimization problems have a solution.} 

\begin{proposition} \label{vorgeschmack} Under Assumption \ref{suposUasymm} the following variational representations hold:
\begin{align}
U_t^p(X)&= \essinf_{Z\in \mathcal{W}^A_t} \left\{\E[ZX\vert \F^A_t] + \alpha^p_t(Z)  \right\} \mbox{  for }X\in L^1_{\F^A_t}(\F^A_{t+1}) \label{repvariacional}, \\
U_t^a(X)&= \essinf_{Z\in \mathcal{W}_t} \left\{\E[ZX\vert \F_t] + \alpha^a_t(Z)  \right\} \hspace{6pt}\mbox{ for }X\in L^1_{\F_t}(\F_{t+1}), \label{repvariacional2}
\end{align}
where 
$$\alpha^p_t:L^{\infty}_{\F^A_t}(\F^A_{t+1})\to \overline{L}^0(\F^A_t)    \;\;\;\; \alpha^a_t:L^{\infty}_{\F_t}(\F_{t+1})\to \overline{L}^0(\F_t),$$
are the respective conjugates of the utility functionals $U_t^p$ and $U_t^a$, and 
$$\mathcal{W}^A_t:=\left\{Z\in L^{\infty}_{\F^A_t}(\F^A_{t+1}):Z\geq 0, \E[Z\vert\F^A_t]=1\right\}\;\;\;\; \mathcal{W}_t:=\left\{Z\in L^{\infty}_{\F_t}(\F_{t+1}):Z\geq 0, \E[Z\vert\F_t]=1\right\}.$$ 
\end{proposition}


{The next example shows that entropic families and pastings of many OCE, such as the TVAR families, fulfilll Assumption \ref{suposUasymm} along with the usual conditions. These are hence the canonical utilities to which our main results in Section \ref{main_results} apply.} 

\begin{example}\label{ejemplificar}
{The entropic families of Example \ref{ejementropic} clearly fulfilll Assumptions \ref{suposUasymm}. In Lemma \ref{ejemx2} we prove that the TVAR families of Example \ref{ejemTVAR}, and more generally the OCE families of Example \ref{oce} for which $1\in int(dom(H^*_t))$ and $H_t$ is bounded from below, fulfilll Assumptions \ref{suposUasymm} after pasting.} 
In Remark \ref{positiva} we will justify that if the Predictable Representation Property holds, then any OCE satisfies Assumption \ref{suposUasymm}.
\end{example}

The variational representation of preferences yields a convenient way to define preference functionals that satisfy Assumption \ref{suposUasymm} by specifying families of ``conditionally acceptable models'' for both parties. 

\begin{example}
For a filtration $\{\mathcal{G}_t\}$ let $\mathcal{A}_t \subset L^{\infty}_{\mathcal{G}_t}(\mathcal{G}_{t+1})$ be a convex set (of conditionally acceptable models) and let $\chi_{\mathcal{A}_t}$ be the associated convex indicator function. Then, the preference functional defined by
\[
	U_t(X) := \essinf_{Z\in \mathcal{W}_t} \left\{\E[ZX\vert \mathcal{G}_t] - \chi_{\mathcal{A}_t}(X) \right\}
\]
satisfies Assumption \ref{suposUasymm} after pasting where $\mathcal{W}_t:=\left\{Z\in L^{\infty}_{\mathcal{G}_t}(\mathcal{G}_{t+1}):Z\geq 0, \E[Z\vert\mathcal{G}_t]=1\right\}$
\end{example}

%
%

\subsection{Contracts and optimal actions}
\label{constracts}
\quad

The simplest contracts the Principal may offer the Agent consist of a fixed $\F_T$-measurable (lump-sum) payment $\Theta$, which we may interpret as a financial derivative contingent only on the path of the price process, plus a constant $\beta$ times $W^A_T$. Such contracts (or more exactly, menus of payments) take the form:
$$\bar{S} = \left\{A\mapsto \bar{S}(A) := \Theta(P_{0:T})+ \beta W^A_T \right\}.$$ 
%

%
Because the Principal observes the {wealth and price processes} progressively, we shall actually consider a wider family of contracts of the form: 
$$S = \left\{A\mapsto S(A) := \Theta(P_{0:T})+   \sum_{t<T} \beta_t\Delta W^A_{t+1}   \right\},$$  
where $\beta_t\in {L^0(\F^A_t)}$ and $\Theta$ is as before, which make better use of her available information. {This contract space is rather large and contains replicable path-dependent derivatives on the wealth process.} However, as in \cite{Ou-Yang}, we shall find that an optimal incentive-compatible contract is indeed of the form $\bar{S}$. This is a consequence of our implicit modeling assumption that the Principal does not seek to infer anything about $A$ from observing $P$ and $W^A$, which we may justify as it being too expensive or time-consuming for the Principal.  
%

We will conveniently refer to a contract as $S$, $(\Theta,\beta)$ or $(\Theta,\{\beta_t\})$ depending on the context and denote by $R\in\R$ the Agent's reservation utility, i.e.\ the least utility the Agent demands in order to commit to a contract 
%

\begin{definition}
\label{incentivecompatible}
A contract $(\Theta,\{\beta_t\})$ is \textit{individually rational} if the optimal utility the Agent can obtain at time $0$ from it is at least $R$. 
%
%
\end{definition}



In the sequel we show how to obtain recursive representations of the Agent's and the Principal's utilities and how to reduce the problem of optimal dynamic contract design to a sequence of static problems.

\subsubsection{Agent's problem}

Let us assume that the Agent chooses an effort level $A$ when presented with a contract $S(\cdot)$. His total cost of effort is then $C(A) := \sum_{t=0}^{T-1}c_t(A_{t})$ and his utility seen from time $t$ is $U^a_t(S(A)-C(A))$. Using translation invariance we compute:
\begin{align}
 U_t^a\left(S(A) - \sum_t c_t(A_t) \right) =&\hspace{3pt}  U_t^a\left( \Theta(P_{0:T}) +  \sum_{s\geq t}\left\{ \beta_s\Delta W^A_{s+1} -c_{s+1}(A_{s+1}) \right\} \right)\notag  \\
 & - c_t(A_t)  +\sum_{s<t} \left\{\beta_s\Delta W^A_{s+1} -c_s(A_s) \right\}. \label{eqlarga}
\end{align}
This shows that the Agent's optimization problem of finding the best effort level $A$ given a contract $S(\cdot)$ reduces to the following recursion (we omit for simplicity the dependence of $H$ in $S$):
\begin{equation} \label{recHbas}
\begin{split}
H_T &= \Theta(P_{0:T})  \\
H_t &=\esssup \limits_{A \in L^0(\F_t)^N} \left\{ U_t^a\left( H_{t+1} + \beta_t A \Delta\tilde{P}_{t+1} \right)  - c_t(A)   \right\}.
\end{split}
\end{equation}

\begin{remark}
The preceding analysis shows that $H_t$ has the interpretation of being the maximal utility the Agent can get, from time $t$ onwards. Since adding an $\F_t-$measurable term to $\Theta$ translates additively into $H_t$ and preserves optimality of effort levels, we see that the individual rationality condition binds ($H_0 = R$) for any contract that is optimal for the Principal.  
\end{remark}

\begin{definition}
A contract $(\Theta,\{\beta_t\})$ is called \textit{incentive-compatible} if the essential suprema in (\ref{recHbas}) are attained for each $t \in \T$.  
%
%
\end{definition}

\subsubsection{Principal's problem} 

The Principal's problem is to design an optimal incentive compatible and individually rational contract. To that end, suppose again that the Agent has chosen $A$ when presented with a contract $S(\cdot)$, and that the Principal knows this. Her utility seen from time $t$ is then:
\begin{align*}
&U_t^p\left (W_T^A-\Theta-\sum_{s<T} \beta_s\Delta W_{s+1}^A \right )\\
&= {W_0^A}- H_t+\sum_{s<t} (1-\beta_s)A_s\Delta\tilde{P}_{s+1} 
+ U_t^p\left( \sum_{s\geq t} \left[(1-\beta_s)A_s\Delta\tilde{P}_{s+1} -\Delta H_{s+1} \right] \right),
\end{align*}
where the identity $\Theta=H_t + \sum_{s\geq t }\Delta H_{s+1}$ and translation invariance was used. If we denote by $h_{t}(A,\beta)$ her utility from future income, then time consistency along with translation invariance yields:
\begin{equation}  \label{intermh}
\begin{split}
	h_{t}(A,\beta) &:= U_{t}^p \left(\sum_{s\geq t} \left[(1-\beta_s)A_s\Delta\tilde{P}_{s+1} -\Delta H_{s+1} \right] \right) \\
	& = U_t^a\left( H_{t+1}+ \beta_t A_t\Delta\tilde{P}_{t+1} \right) - c_t(A_t)  +
U_t^p\left(h_{t+1}(A,\beta) + (1- \beta_t) A_t\Delta\tilde{P}_{t+1} - H_{t+1}   \right) .
\end{split}
\end{equation}
%
%

Performing the change of variables 
%
\begin{equation}
\Gamma_{t+1} := \beta_t A_t\Delta\tilde{P}_{t+1} + H_{t+1}  \in L^0(\F_{t+1}),\label{covG}
\end{equation}
%
and writing $h_t(A,\Gamma)$ instead of $h_t(A,\beta)$ we arrive at:
\begin{equation}
h_t(A,\Gamma) =  U_t^a( \Gamma_{t+1} ) - c_t(A_t)  +
U_t^p\left(h_{t+1}(A,\Gamma) +   A_t\Delta\tilde{P}_{t+1} - \Gamma_{t+1}    \right).
\label{intermhG}
\end{equation} 

%

%

If $(\Theta,\{\beta_t\})$ is incentive compatible, then unique optimal effort levels for the Agent exist, due to our concavity assumptions on his utility and cost function. For every time $t \in \T$ we 
may thus construct  the random variable $\Gamma_{t+1}$, and $A_t$ will attain the essential supremum: 
$$\esssup_{a}\left[ U^a_t\left(\Gamma_{t+1} + \beta_t[ a - A_t ]\Delta\tilde{P}_{t+1}\right) -c_t(a)\right ] .$$ 
We say that $(\{A\},\{\Gamma\})$ is \textit{incentive-compatible} whenever for every $t$ this $A_t$ attains this supremum. In terms of the set 
$$\C_t(\beta):= 
\left\{
\begin{array}{c}
(A,\Gamma)\in [L^0(\F_t)]^N\times L^0(\F_{t+1})\mbox{ s.t. for every } \bar{A}\in [L^0(\F_t)]^N: \\
U^a_t(\Gamma) -c_t(A) \geq U^a_t\left(\Gamma +\beta [ \bar{A} -A ]\Delta\tilde{P}_{t+1}\right) -c_t(\bar{A})
\end{array}
\right\},$$
incentive compatibility amounts to $(A_t,\Gamma_{t+1})\in \C_t(\beta_t)$ for every $t \in \T$. In particular, we can introduce the following recursion for the Principal's future optimal wealth:
%
\begin{equation} \label{rechas}
\begin{split}
h_T&=0, \\
h_t&= \esssup\limits_{ \substack{{(\beta,A,\Gamma)} \\ {(A,\Gamma)\in \C_t(\beta)}}}  U_t^a( \Gamma ) - c_t(A)  + U_t^p\left(h_{t+1}+   A\Delta\tilde{P}_{t+1} - \Gamma   \right). 
\end{split}
\end{equation}
%


\begin{remark}\label{blabla}
We arrived at the well-known result that in constructing an optimal contract the Principal should consider the Agent's continuation utility as a decision variable of hers. This also resolves the issue of information asymmetry: assuming that the Principal knows the mappings $A_t$ as functions of $\{P_s\}_{s\leq t}$ for each $t$ implies that all the random variables in \eqref{intermh} and \eqref{intermhG} become price-adapted.\footnote{We emphasize, again, that this is a consequence of the assumption that the Principal is not trying to learn/infer something from the Agent's actions.} If optimal efforts are not unique, then one has to specify which effort levels $\{A_t\}$ (the Principal recommends) the Agent implements in order to carry out the above recursion. This is why in the PA literature one often calls such effort levels {\sl recommended effort levels} and the triple $(\Theta,\{\beta_t\},\{A_t\})$ incentive compatible.
\end{remark}

\subsection{Main results}\label{main_results}
\quad

In this section we summarize the main results of our paper. We start with the following theorem that makes our formal derivations of the Agent's and Principal's optimal wealth precise. 
{It states that if the Principal's and the Agent's conditional one-step optimization problems have solutions, then the dynamic contracting problem has a solution that can be obtained out of these.} 
The proof is given in Appendix C.

\begin{theorem}
\label{master}
Assume that the recursions \eqref{recHbas} and \eqref{rechas} admit a solution with the essential suprema attained at each time $t$. Then the Principal's optimal utility at time $t=0$ equals $W_0-R+h_0$. Further, calling $(\beta_t,A_t,\Gamma_{t+1})_{t\in\T}$ the maximizers attaining $h$ in \eqref{rechas}, and defining 
%
$$\Theta =\Theta(P_{0:T}):= \sum_{0\leq t< T} \left[ \Gamma_{t+1} - U^a_t(\Gamma_{t+1}) + c_t(A_t) \right],$$ 
the contract 
$$S=\left\{\bar{A}\mapsto R + \Theta(P_{0:T})+\sum \beta_t\left [\Delta W_{t+1}^{\bar{A}}{-A_{t}\Delta\tilde{P}_{t+1}}\right]\right\},$$
 is optimal for the Principal, among those satisfying incentive compatibility and individual rationality constraints. The associated optimal effort for the Agent is  $A$ and his optimal wealth will be $R$. 
\end{theorem}



We now define an auxiliary unconstrained version of the optimization problem in \eqref{rechas}, and prove that if such a problem is well-posed, it yields a at time $t$ a solution to the original one-step problem, and the corresponding $\beta_t = 1$ is optimal. {This opens the way to our main result, Theorem \ref{casodina}}. The technical importance of this is that we may dispense with the non-convex sets $\C_t$, making the incentive-compatibility constraint much more tractable. {Economically, this indicates that the first-best solution is optimal if it exists: 
at any point in time $t \in \T$ both parties share the ``aggregate endowment'' given by the Principal's utility from future income $h_{t+1}$ plus gains from trading $A \Delta \tilde{P}_{t+1}$ so as to maximize aggregate utility. }

\begin{proposition}
\label{reduc}
Assume that the following problem is finite and attainable:
\begin{equation}
\label{PI}
\Sigma\hspace{3pt}:=\esssup\limits_{(A,\Gamma)\in [L^0(\F_t)]^N\times L^0(\F_{t+1}) }  U_t^a( \Gamma ) - c_t(A)  + U_t^p\left(h_{t+1}+  A\Delta\tilde{P}_{t+1} - \Gamma   \right).
\end{equation}
Then any maximizer $\left(\hat{A},\hat{\Gamma}\right)$ belongs to the set $\C_t(1)$ and therefore 
\begin{align*} 
\Sigma \,\,\, &=  \esssup\limits_{ \substack{{(\beta,A,\Gamma)} \\ {(A,\Gamma)\in \C_t(\beta)}}} U_t^a( \Gamma ) - c_t(A)  + U_t^p\left({h_{t+1}}+   A\Delta\tilde{P}_{t+1} - \Gamma    \right).
\end{align*}

\end{proposition}

\begin{proof}
Let $(\hat{A},\hat{\Gamma})$ be a maximizer for \eqref{PI}. For arbitrary $A$, define $\Gamma=\hat{\Gamma}+(A-\hat{A})\Delta\tilde{P}$. Plugging in that $(\hat{A},\hat{\Gamma})$ is better than $(A,\Gamma)$ for \eqref{PI}, we see that the terms involving $U^p$ cancel out and so: 
\begin{equation} 
U^a_t(\hat{\Gamma})-c_t(\hat{A}) \geq U^a_t(\hat{\Gamma}+(A-\hat{A})\Delta\tilde{P})-c_t(A). \label{interm}
\end{equation}
This means that $\left(\hat{A},\hat{\Gamma}\right)\in\C_t(1)$ so the values of the constrained and unconstrained problems coincide.
\end{proof}

\begin{remark}
The previous proof crucially relies on the fact that contracts are {linear in wealth increments}. Indeed by varying $\hat{\Gamma}$ in directions of the form $(A-\hat{A})\Delta\tilde{P}$ and by linearity of contracts the term in the objective function involving Principal's utility cancel out, making it possible to compare the values of Agent's utilities. 
%
\end{remark}
In Section \ref{GenResL0} we shall, therefore, turn our attention to the question of attainability of the unconstrained problem. For the reader's convenience we state in this section our main results therein and show how they apply to specific classes of examples. The proof of the following result is given in Section \ref{PrinProb}. {The technical conditions will be easily satisfied by the {utility functionals listed in Example \ref{ejemplificar}.}}

%


\begin{theorem}
\label{teogeralasim}
Suppose at time $t\in\T$ that
%
%
%
$$K^{p}_t:=\esssup\limits_{Z\in \mathcal{W}_t\cap [1-\epsilon,1+\epsilon]}\alpha_t^{p}(Z) \in L^0(\F_t) \hspace{10pt}\mbox{ and }\hspace{10pt}K^{a}_t:=\esssup\limits_{Z\in \mathcal{W}_t\cap [1-\epsilon,1+\epsilon]}\alpha_t^{a}(Z) \in L^0(\F_t),$$
 for some $\epsilon\in L^0(\F_t)\cap(0,1]$. Then, if $h_{t+1}\in dom(U^p_t)$ and $\lim\limits_{\vert a \vert \rightarrow \infty}\frac{c_t(a)}{\vert a\vert}=+\infty$, 
the random variable {$\Sigma$ defined in \eqref{PI} belongs to $L^0(\F_t)$, satisfies}
$$\Sigma =\esssup\limits_{(A,\Gamma)\in [L^0(\F_t)]^N\times L^1_{\F_{t}}(\F_{t+1}) }   U_t^a( \Gamma ) - c_t(A)  + U_t^p\left(h_{t+1}+  A\Delta\tilde{P}_{t+1} - \Gamma   \right)
,$$ 
and the essential supremum is attained. In particular $\beta_t = 1$ is optimal at time $t \in\T$.
\end{theorem}

It is well-known that if the utility functionals originate from a common base functional, more explicit treatments of equilibrium/risk-sharing problems become available (as in \cite{BarKa,Borch,Equilibrium}). In the same spirit we have the following result, {stating that in that case the Principal and the Agent share the ``aggregate endowment''  according to their risk attitudes.} The proof is given in Section \ref{PrinProb}. 

\begin{theorem}[Base Preferences]
\label{teoUsimilar} 
Suppose that there exists non-negative numbers $\gamma^a,\gamma^p$ and base preference functionals $\{U_t\}$ such that 
$$U^l_t(\cdot) := \frac{1}{\gamma^l} U_t\left(\gamma^l \cdot \right) \quad (l=a,p).$$
Further assume that 
$$\frac{\gamma^a\gamma^p}{\gamma^a+\gamma^p}h_{t+1}\in dom(U_t) \hspace{6pt}\mbox{ and }\hspace{6pt} \lim\limits_{\vert a \vert \rightarrow \infty}\frac{c_t(a)}{\vert a\vert}=+\infty.$$
Then Principal's one-step problem (at time $t$) has as solution: $$\beta =  1 \mbox{ and }\Gamma^* = \frac{\gamma^p}{\gamma^a+\gamma^p} (h_{t+1}+A^*\Delta \tilde{P}_{t+1}),$$ for the optimal action $A^*$ of the Agent, which attains:
$$\esssup_A \left\{ -c_t(A) + \frac{\gamma^a+\gamma^p}{\gamma^a\gamma^p}U_t\left(\frac{\{\gamma^a\gamma^p\}[h_{t+1}+A\Delta \tilde{P}_{t+1}]}{\gamma^a+\gamma^p}\right) \right\} .$$

\end{theorem}

In light of Theorem \ref{master}, the two previous results yield a solution to the dynamic problem, as explained in the following proposition. 

\begin{proposition}
\label{Genmetateo}
If the assumptions of Theorem \ref{teogeralasim} or Theorem \ref{teoUsimilar} hold for every $t \in\T$,
then the respective one-step problems have a solution and glueing them together yields a solution for the respective dynamic problems, whereby $\beta_t = 1$ for every $t \in\T$ is optimal. 
\end{proposition}


{The proof of the preceding proposition is obvious. In applying this result, several technical conditions need be checked {\sl a-posteriori}. As shown by the following theorem, these conditions are satisfied {\sl a-priori} for entropic and TVAR families and for OCE utilities (Example \ref{ejemplificar}) under mild conditions.} The proof is given in Appendix B.

\begin{theorem}
\label{casodina}
Suppose that prices are bounded ($0<p_-\leq P_t^i\leq p_+$ a.s.) and that both $U^a$ and $U^p$ are constructed by pasting of optimized certainty equivalent functionals:
\begin{align*}
X\in L^1_{\F_{t}}(\F_{t+1})&\mapsto U^a_t(X)=\esssup_{s\in\R}\{s-\E[H^a_t(s-X)|\F_t]\} \\
X\in L^1_{\F^A_{t}}(\F^A_{t+1})&\mapsto U^p_t(X)=\esssup_{s\in\R}\{s-\E[H^p_t(s-X)|\F^A_t]\},
\end{align*} 
{for which the following conditions hold for each $t$:
\begin{itemize}
\item $1\in int(dom(H^a_t))\cap int(dom(H^p_t))$, 
\item $H^a_t$ and $H^p_t$ are lower-bounded.
\end{itemize}}
Finally assume that $\lim_{|a|\to\infty}\frac{c_t(a)}{|a|}=\infty$ for every $t$. 
Then our dynamic Principal-Agent problem has a solution whereby the Agent keeps the output wealth and the Principal is given a possibly path-dependent derivative. 

\end{theorem}

\begin{remark}
{In conjunction with Theorem \ref{master}, the previous result yields the economic interpretation we referred to in the introduction: the optimal contract is of the form ``cash plus a path-dependent derivative on the stock price process plus {performance w.r.t.\ }a benchmark portfolio''. As the derivative may not be replicable, this shows that the structure in \cite[Theorem 1]{Ou-Yang} need not hold. }
\end{remark}

A family of examples where we can provide explicitly the form of an optimal contract, recovering the results of \cite{Ou-Yang} in the continuous case is given in Section \ref{LagLag} below. It requires additional notation, though, so we postpone the statement of the result to Section 4. 

{
\begin{remark}\label{rem:simplification}
For simplicity and ease of exposition we took zero interest rates and $c_t=c_t(A_t)$. The case with non-null interest rates and/or $c(A,W)=\sum_t[c_t(A_t)+\gamma_t W_{t-1}]$ can be solved exactly in the same way, the only difference being that $\beta_t$ will not be constant (but remains deterministic) anymore. The qualitative structure of contracts and their interpretation remain the same however.
\end{remark}
}  

%

\section{General attainability results}
\label{GenResL0}

We prove in this section the attainability of the Agent's and Principal's one-step problems, and consequently, for the dynamic problem. 


\subsection{Agent's Problem}
\quad

We start with an abstract conditional optimization problem of which the Agent's one-step optimization problems are special cases. For a given pair of random variables $(X,\beta)\in L^0(\F_{t+1})\times L^0(\F_{t})$, let   
\begin{equation} \label{sup-problem}
\begin{split}
	G(t,X,\beta) &:=  \esssup\limits_{A\in L^0(\F_t)^N}\left\{-c_t(A) + U_t^a\left(X+ \beta A\Delta \tilde{P}_{t+1}  	\right) \right\} \\
	& =:   \esssup\limits_{A\in L^0(\F_t)^N} g_t(A) .
\end{split}
\end{equation}
Under the usual conditions $g_t$ is $\F_t-$concave, and hence stable (see Definition \ref{directorio2}). The key to the above optimization problem is to reduce it to an $L^0(\F_t)-$bounded set. 

\begin{lemma}
\label{attaing} 
Under the following condition, the essential supremum in (\ref{sup-problem}) is attained:
$$X \in  dom(U^a_t) \hspace{10pt}
\mbox{ and } \hspace{10pt}\lim\limits_{\vert a \vert \rightarrow \infty}\frac{c_t(a)}{\vert a\vert}=+\infty.$$ 
\end{lemma}
  
\begin{proof}
We intend to apply Theorem \ref{CondOpt}. Evidently $$\esssup_{A\in [L^0(\F)]^N} g_t(A) =\esssup_{A \in\Lambda} g_t(A),$$ where $\Lambda=\{A:g_t(A)\geq g_t(0)\}$. The set $\Lambda$ is $L^0-$convex, contains the origin and is $\sigma-$stable. That $\Lambda$ is sequentially closed is an application of Proposition \ref{implicanciausc}.

For $A \in [L^0(\F_t)]^N$ not identically null we use the variational representation of $U^a$ established in Proposition \ref{vorgeschmack} to bound:
\begin{equation} \label{equtil}
\begin{split}
g_t(nA)&=U_t^a(X+n\beta A\Delta \tilde{P})-c_t(nA) \\
&\leq K+ \E[ZX_+\vert\F_t] + n\E[\beta Z A\Delta \tilde{P}\vert\F_t]-c_t(nA),
\end{split}
\end{equation}
where $Z\in \mathcal{W}_t$. 
Using that $A,\beta$ are $\F_t$-measurable and Cauchy-Schwarz applied pointwise, we bound from above the sum of the last two terms in \eqref{equtil} on the set where $A$ does not vanish by:
$$n\vert A\vert \vert\beta\vert \vert \E[Z\Delta\tilde{P}\vert\F_t] \vert   -c_t(nA) \leq  n\vert A\vert \left [\vert\beta\vert \vert \E[Z\Delta\tilde{P}\vert\F_t] \vert - \frac{c_t(nA)}{n\vert A\vert} \right ].$$ 
%
Since $Z\in L^{\infty}_{\F_t}(\F_{t+1})$ and $\vert \E[\Delta\tilde{P}\vert\F_t] \vert$ is a.s.\ finite by assumption, we see that the majorizing term tends to $-\infty$ on a non-negligible set as $n \to \infty$ and so does $g_t(nA)$. Since $g_t(0)=U^a_t(X)-c_t(0) >-\infty$ by assumption, we get a contradiction, and so Theorem \ref{L0convbound} shows that $\Lambda$ is $L^0(\F_t)$-bounded. Hence Theorem \ref{CondOpt} applies to $\esssup_{A \in\Lambda }g_t(A)$, since the mapping $A \mapsto g_t(A)$ is $L^0$-upper semicontinuous by Proposition \ref{implicanciausc}. This establishes attainability.
\end{proof}



The following is an immediate corollary of the previous lemma.

\begin{corollary}
Assume that $H_{t+1} \in  dom(U^a_t)$ and $\lim\limits_{\vert a \vert \rightarrow \infty}\frac{c_t(a)}{\vert a\vert}=+\infty$. Then the one-step conditional optimization problem of the Agent at time $t$, as in \eqref{recHbas}, is attained.

\end{corollary}

\subsection{Principal's Problem}\label{PrinProb}
\quad

In this section we prove Theorem \ref{teogeralasim}, which sharpens Proposition \ref{reduc}. The Principal's problem at time $t$ consists in maximizing 
$$V_t(A,\Gamma):= U_t^a( \Gamma ) - c_t(A)  + U_t^p\left(h_{t+1}+  A\Delta\tilde{P}_{t+1} - \Gamma   \right). $$
Recall from Remark \ref{blabla} that the Principal's preference functionals $U^p_t$ may and will be considered as mappings from $L^0(\F_T)$ to $L^0(\F_t)$, satisfying the usual assumptions w.r.t.\ $\F$. 
%

\mbox{ } 

{\sl Proof of Theorem \ref{teogeralasim}.} Let us introduce the set
\[
	\mathcal{S}:=\{(A,\Gamma)\in L^0(\F_t)^N\times Q: V(A,\Gamma)\geq V(0,0)\}, 
\]	
where $Q:=\left\{\Gamma \in L^1_{\F_{t}}(\F_{t+1}):\E[\Gamma \vert \F_t]=0\right\}$.
{In maximizing $V$, i.e.\ in computing $\Sigma$, we may assume that $\Gamma\in L^1_{\F_{t}}(\F_{t+1})$, since for candidate optima, $\Gamma$ and $-\Gamma$ must be in the domains of $U^a$ and $U^p$ respectively, and by assumption this yields $\Gamma_{-},\Gamma_+ \in L^1_{\F_{t}}(\F_{t+1})$.} We may thus further assume that $(A,\Gamma)$ belong to $\mathcal{S}$, since $\F_t$-measurable components of $\Gamma$ cancel out in $V$, i.e.
\[
	\Sigma=\esssup_{(A,\Gamma)\in\mathcal{S} }V(A,\Gamma).
\]
In a first step, we will show that the set 
\[
	\mathcal{S}^A:=\{A\in L^0(\F_t)^N: \mbox{ there exists } \Gamma\in Q \mbox{ such that }(A,\Gamma)\in\mathcal{S}\}
\]	
is $L^0(\F_t)-$bounded. To this end, we first notice that $V(0,0)=-c(0) + U^p_t(h_{t+1}) \in L^0(\F_t)$. Taking 
\[
	\tilde{Z}\in dom(\alpha^p)\cap dom(\alpha^a)\cap\mathcal{W}\cap L^0(\F_t)
\]	
(e.g.\ $\tilde{Z}=1$) and using the variational representation of the preference functionals, we get: 
\begin{eqnarray*}
	U^a_t(\Gamma) &\leq& \alpha^a(\tilde{Z})+\tilde{Z}\E[\Gamma\vert\F_t] \\
	U^p_t(h + A \Delta \tilde{P}-\Gamma ) &\leq& \alpha^p(\tilde{Z})+ \tilde{Z}\E[h\vert\F_t]-\tilde{Z}\E[\Gamma\vert\F_t] + \tilde{Z}\E[A\Delta \tilde{P}_{t+1}\vert \F_t].
\end{eqnarray*}
For $\Gamma \in Q$ the term $\E[\Gamma\vert\F_t]$ vanishes and hence
\small $$V(0,0)\leq \alpha^p(\tilde{Z})+\alpha^a(\tilde{Z})+\tilde{Z}\E[h\vert\F_t] + \vert A\vert   \vert\tilde{Z}\E[ \Delta \tilde{P}_{t+1} \vert\F_t]\vert -c(A).$$ 
\normalsize Since $\mathcal{S}^A$ is $\sigma$-stable, we can use Lemma \ref{lemnoacotado} to conclude. Indeed, if $\mathcal{S}^A$ were not $L^0(\F_t)$-bounded, then there exists a non-negligible set $\tilde{\Omega}$ and a sequence $\{A_n\}\subset \mathcal{S}^A$ such that $\vert A_n\vert \geq n$ on $\tilde{\Omega}$. Similar arguments as in the proof of Lemma \ref{attaing} would establish $V(0,0)=-\infty$ on a non-negligible set, contradicting our hypotheses. Thus $\mathcal{S}^A$ must be $L^0(\F_t)$-bounded.

Next, we prove that the set 
\[
	\mathcal{S}^\Gamma:=\{\Gamma \in Q: \mbox{ there exists } A \in L^0(\F_t)^N \mbox{ such that }(A,\Gamma)\in\mathcal{S}\}
\]
is bounded in $L^1_{\F_t}(\F_{T})$. Let us chose $\epsilon\in L^0(\F_t)\cap (0,1]$ as in the statement of this theorem, fix $\Gamma \in \mathcal{S}^\Gamma$ and define 
\begin{eqnarray*}
	Z^a &:=& 1+\epsilon[\ind{\Gamma\leq 0} - \Prob(\Gamma\leq 0\vert\F_t)] \in L^{\infty}(\F_T)\cap[1-\epsilon,1+\epsilon] \\ 
	Z^p &:=& 1+\epsilon[\ind{\Gamma> 0} - \Prob(\Gamma> 0\vert\F_t)]  \in L^{\infty}(\F_T)\cap[1-\epsilon,1+\epsilon]
\end{eqnarray*}
%
Since $\Gamma\in Q$ we see that
 $$\E[Z^a\Gamma\vert\F_t]=-\epsilon\E[(\Gamma)_{-}\vert\F_t]\hspace{10pt}\mbox{ and }\hspace{10pt}\E[Z^p\Gamma \vert\F_t]=\epsilon\E[(\Gamma )_+ \vert\F_t] .$$
Moreover, $\E[Z^{a}\vert\F_t]=\E[Z^{p}\vert\F_t]=1$, implying that $Z^{a,p}\in \mathcal{W}_t$ and thus $\alpha^{p}(Z^{p})\leq K^{p}$ and $\alpha^{a}(Z^{a})\leq K^{a}$. We hence obtain that
\begin{align*}
U^a(\Gamma) &\leq -\epsilon\E[(\Gamma)_{-}\vert\F_t] + K^a, \\
U^p(h+A \Delta \tilde{P}-\Gamma)&\leq \E[Z^p(h+A \Delta \tilde{P})\vert\F_t]-\epsilon\E[\Gamma_+\vert\F_t]+K^p \\
&\leq 2\E[\vert h\vert \vert\F_t]+2\vert A\vert\E[  \vert\Delta \tilde{P} \vert\vert\F_t]-\epsilon\E[\Gamma_+\vert\F_t]+K^p \\
&\leq N-\epsilon\E[\Gamma_+\vert\F_t],
\end{align*}
for some $N\in L^{0}(\F_t)$ where the latter inequality follows by assumption and the fact that the effort levels had already been proven to be $L^0(\F_t)$-bounded. Therefore for $(A,\Gamma)\in\mathcal{S}$ we have 
$$V(0,0)\leq V(A,\Gamma)\leq N+K^a-\epsilon E[ (\Gamma)_-\vert\F_t]-\epsilon\E[\Gamma_+\vert\F_t]-c_t(A) \leq \tilde{K} -\epsilon E[ \vert\Gamma\vert\vert\F_t],$$ 
for some $\tilde{K}\in L^{0}(\F_t) $. This implies that $\mathcal{S}^{\Gamma}$ is bounded in $L^1_{\F_t}(\F_{T})$ since $\epsilon >0$ a.s.


Next, we notice that there exists a sequence $(A_n,\Gamma_n)\in \mathcal{S}$ such that $V(A_n,\Gamma_n) \uparrow \Sigma$ since $\mathcal{S}$ is directed upwards. Indeed, if $V(A^i,\Gamma^i)\geq V(0,0)$ for $i=1,2$ and if we define $\xi=\{V(A^1,\Gamma^1)\geq V(A^2,\Gamma^2)\}$ and $(A,\Gamma)=(A^1,\Gamma^1)\ind{\xi}+(A^2,\Gamma^2)\ind{\xi^c}$, then  
$$V(A,\Gamma)=\max\{V(A^1,\Gamma^1),V(A^2,\Gamma^2)\}\geq V(0,0),$$
thanks to the terms in $V$ being $\F_t-$stable and $\xi\in \F_t$. 
%
%
{By virtue of $\mathcal{S}^A$ being $L^0(\F_t)$-bounded, we can apply the usual Komlos lemma (or Lemma \ref{teoDSenchulado}) to the positive and negative parts of each component of the sequence $\{A_n\}_n$ in an iterative, nested way, i.e.\ taking convex combinations of convex combinations and so forth. On the other hand, the $L^1_{\F_t}(\F_{T})$-boundedness of $\mathcal{S}^{\Gamma}$ implies that the technical condition in Lemma \ref{teoDSenchulado} holds for the positive and negative parts of the sequence $\{\Gamma_n\}_n$, by Jensen's inequality, so we can again take convex combinations of convex combinations. All in all we have found a sequence} of non-negative real numbers $\{\lambda_i^n\}$ with $\sum_{i\geq n}\lambda_i^n=1$, and random variables $\Gamma^*\in L^0(\F_{t+1})$ and $A^*\in L^0(\F_t)^N$ such that $\tilde{\Gamma}_n= \sum_{i\geq n} \lambda^n_i \Gamma_i \rightarrow \Gamma^*$ and $\tilde{A}_n= \sum_{i\geq n} \lambda^n_i A_i \rightarrow A^*$ a.s.\ {(for each component)}. Also $(\tilde{A}_n,\tilde{\Gamma}_n)\in \mathcal{S}$ by convexity. Moreover,  
$$\Sigma=\lim_{n} V(A_n,\Gamma_n)= \lim_n \sum_{i\geq n}\lambda^n_i V(A_i,\Gamma_i)\leq \limsup_n V(\tilde{A}_n,\tilde{\Gamma}_n),$$
since (a.s.) convergent sequences of real numbers remain converging under convex combinations of its tails and $V$ is concave.
 %
%

{The cost-term in $V$ is u.s.c.\ and since $\mathcal{S}^{\Gamma}$ is $L^1_{\F_t}(\F_{T})$-bounded we get for the $U^a$ term in $V$ that $\limsup_n U^a\left (\tilde{\Gamma}_n\right ) \leq U^a(\Gamma^*)$. Finally, for the $U^p$ term in $V$, we obtain from the last assertion in  Proposition \ref{implicanciausc} that $\limsup_n U^p\left (h_{t+1}+\tilde{A}_n\Delta \tilde{P}_{t+1}-\tilde{\Gamma}_n\right ) \leq U^p(h_{t+1}+A^*\Delta\tilde{P}_{t+1}-\Gamma^*)$.} 
We thus get that $\Sigma\leq V(A^*,\Gamma^*)$ and hence we have equality. This shows that $\Sigma <\infty$ since the preference functionals are proper. Finally, by Proposition \ref{reduc} we conclude that $\beta =1$ is optimal and Principal's one-step problem is attained. 
\hfill $\Box$


%
%

%

We proceed now to the proof of Theorem \ref{teoUsimilar}. 

\mbox{ } 

{\sl Proof of Theorem \ref{teoUsimilar}.}  
%
Let us first fix an effort level $A$ and put $x:=h+A\Delta \tilde{P}_{t+1}$ and $\hat{\gamma}:= \frac{\gamma^a\gamma^p}{\gamma^a+\gamma^p}$. Concavity of the preference functional yields:
\begin{eqnarray*}
	\esssup_{\Gamma}  \left\{ U_t^a(\Gamma)+U_t^p(x-\Gamma) \right\}
	= \esssup_{\Gamma}  \frac{1}{\hat{\gamma}}\left\{ \frac{\hat{\gamma}}{\gamma^a}U_t(\gamma^a \Gamma)+ \frac{\hat{\gamma}}{\gamma^p}U_t(\gamma^p [x-\Gamma])\right\} \leq \frac{1}{\hat{\gamma}} U_t(\hat{\gamma} x). 
\end{eqnarray*}
 On the other hand, taking $\Gamma^*=\frac{\gamma^p}{\gamma^a+\gamma^p}x$ it follows that
$\frac{1}{\gamma^a}U_t(\gamma^a \Gamma^*)+ \frac{1}{\gamma^p}U_t(\gamma^p [x-\Gamma^*])=\frac{1}{\hat{\gamma}} U_t(\hat{\gamma} x)$. Therefore this $\Gamma^*$ attains the essential supremum above. Thus the Principal's problem reduces to:
\begin{equation}
\label{eqw}
	\esssup_{\Gamma} \left\{ -c_t(A) + \frac{1}{\hat{\gamma}}U_t\left(\hat{\gamma}[h_{t+1}+A\Delta \tilde{P}_{t+1}]\right) \right\}.
\end{equation}
If this problem is attained at $A^*$, then the previous argument shows that $\Gamma^* = \frac{\gamma^p}{\gamma^a+\gamma^p} (h_{t+1}+A^*\Delta \tilde{P}_{t+1})$ is optimal. The problem \eqref{eqw} is of the same form as that analyzed in Lemma \ref{attaing}, simply replacing $U^a$ by $\frac{1}{\hat{\gamma}} U_t(\hat{\gamma}\cdot) $, calling $X=h$
and taking $\beta =1$. In particular, we obtain existence of an optimizer $A^*$. Because the one-step unconstrained problem is attained, Proposition \ref{reduc} shows that
taking $\beta = 1$, $A^*$ and $\Gamma^*$ at time $t$ yields an optimal one-step decision. \hfill $\Box$


\begin{remark}
In this article we chose to work in the biggest conditional (loc.\ convex) space of $L^p$-type, this is, the conditional $L^1$ space. The reason is twofold. On the one hand, had we worked with smaller subspaces, we would have had in principle more tools at hand to prove the attainability of Principal's one-step problems. However, we chose not to limit the scope of utility functionals a priori, in terms of their domains, for which the theory would be applicable to. On the other hand, even acknowledging the fact that our $L^0-L^1$ upper semicontinuity requirement is not a mild one, the alternative would have been to impose from the outset some sort of ``sup-compactness'' of our functionals (more precisely, of their convolutions) or again to work with smaller spaces than conditional $L^1$; ideally conditionally reflexive ones. It seems to us that our simple sequential (and rather point-wise) $L^0-L^1$ upper semicontinuity has the advantage of being a more tractable and less technical requirement than the other, very valid ones.
\end{remark}

\section{Optimal contracting under predictable representation}
\label{mark}

Up to now our probability space and price process were rather general. In this section we add more structure to the problem in order to obtain more explicit solutions. In particular we fix a volatility matrix $\sigma \in \R^{N,d}$ with linearly independent rows ($d \geq N$), assume that the flow of information is generated by a d-dimensional process $\bar w = (w^1, ..., w^d)$ whose evolution is observed by both parties and that the price dynamics follows:
\begin{equation}
\Delta P_{t+1}= diag(P_t)\left [ \mu  + \sigma \Delta \bar{w}_{t+1}\right ] \label{dyn}
\end{equation}
%


Moreover, we shall work under the following ``Predictable Representation Property'' and assume that our utility functionals satisfy a Markov condition. 
 
\begin{assumption}
\label{suposPRP} 
%
The \textit{Predictable Representation Property (PRP)} holds: for some $D \in \mathbb{N}\cup\{0\}$ there exists processes $w^{d+1}, ..., w^D$ adapted to the filtration $\{\F_t\}$ generated by the process $\bar w$ such that the extended process $w=(\bar{w}^1, ..., \bar{w}^d, w^{d+1}, ..., w^D)$ has uncorrelated increments which are independent from the past, have zero mean, non-trivial finite second moments, 
and
\begin{equation}
L^0(\F_{t+1}) =  \left\{x + Z  \Delta w_{t+1} \mbox{ : } x\in L^0(\F_t),Z\in [L^0(\F_t)]^D \right\}. \label{PRP}
\end{equation}
%

We stress that if initially the  $d$-dimensional $\bar w$ process driving the price had not enjoyed the PRP, then Assumption \ref{suposPRP} simply says that we can complete the former process in such a way that the enlarged process does enjoy the PRP, without changing the informational structure of the model. The following example clarifies our PRP assumption. 

\begin{example}[Bernoulli Walk]
\label{introbernoulli}
Consider in $\RR^d$, $d$ independent Bernoulli walks $w^1,\dots,w^d$ on the time grid $\{0,h,2h,\dots,T\}$ starting at $0$, such that $\Prob(\Delta w^i_t =\sqrt{h}) = \Prob(\Delta w^i_t = -\sqrt{h}) = \frac{1}{2}$. They do not necessarily fulfilll \eqref{PRP}, unless $d=1$. Yet it is well-known that for $D=2^d-1$, there exists an adapted family $w^{d+1},\cdots,w^{D}$ of likewise distributed random walks, such that the whole extended family $w^1,\dots,w^{D}$ has increments uncorrelated to each other and independent from the past, and such that \eqref{PRP} holds.
\end{example}

We further restrict ourselves to preference functionals which satisfy the following Markov Property.  

\end{assumption}
\begin{assumption}
\label{suposmarkov}
The \textit{generators} $g^l$ (l=a,p) defined by 
$$Z \in [L^0(\F_t)]^D \mapsto g_t^l(Z) := U_t^l( Z \Delta w_{t+1}),$$
are \textit{Markovian} in the sense that $g^a,g^p$ map $\RR^D$ to $\RR$. 
\end{assumption}

If a preference functional $U$ satisfies the usual conditions and the PRP holds, then all the relevant information of $U_t$ is summarized by its generator. Clearly $g_t$ inherits from $U_t$ being null at the origin and concave. In the case that $P$ may only take a finite number of values, and by the ``local property,'' $\ind{Z(\cdot)=z}g_t(Z)(\cdot)=\ind{Z(\cdot)=z}g_t(z)(\cdot) $. 

\begin{example}
\label{ejemutilities}
For optimized certainty equivalents, the generator $g(x) := U_t(x\Delta w_t)=\sup_{s}\{s-\E(H(s-x\Delta w_t))\}$ clearly satisfies the markovianity assumption under the PRP. 
\end{example}




\begin{remark}\label{positiva}
Under Assumption \ref{suposPRP} one could re-write the Agent's and Principal's recursions as Backward Stochastic Difference Equation in a direct way. In doing so we would replace $\Gamma$ by $\gamma\Delta w$ everywhere in Principal's problem, this having major advantages as by-product. First, one may drop the $L^0-L^1$ upper semicontinuity assumption and simply work with the variational representations of the utility functionals. Indeed, by \eqref{bla} of Proposition \ref{implicanciausc} this would imply $L^0$ upper semicontinuity of $V$ (as in Principal's one step unconstrained problems) in the variables $(A,\gamma)$, which is all we need. As a consequence, the results of the previous section extend to e.g.\ every optimized certainty equivalent utility in the PRP case.
We spare the reader the repetitive work of proving the above points, and instead proceed to a more explicit characterization of optimal contracts. 

\end{remark}

From the substitution $\Gamma_{t+1}-\E[\Gamma_{t+1}|\F_{t+1}]=\gamma \Delta w_{t+1}$ for some $\gamma \in [L^0(\F_t)]^D$ valid by the PRP assumption, we may call a tuple $(A,\beta,\gamma)$ without danger of confusion a \textit{contract} (we shall always work with these variables under the PRP). Principal's recursion \eqref{rechas} and the incentive compatibility set $\C_t(\beta)$ may then be re-defined in terms of such tuple in an obvious way. 

\begin{remark}
\label{non-randomn}
From equation \eqref{rechas} it becomes apparent that under the PRP and Markovianity Assumptions $h_t$ becomes a real number for all $t$. Indeed, everything in the one-step optimization problems (the $g$'s and $c$'s) is non-random when evaluated at non-random inputs, from which it suffices to consider $(A,\beta,\gamma)\in\RR^N\times \RR \times \RR^D$ and maximize point-wise. This of course shows that in this case if there is an optimal contract, then the optimizer $(A,\beta,\gamma)$ is non-random.
\end{remark}
%
%
%
%
%
%

\subsection{Computing optimal contract and necessary optimality conditions}
\label{LagLag}
\quad

Starting from the original formulation \eqref{rechas}, we tackle the attainability issue without resorting immediately to the unconstrained variant. We will thus see that in fact solving this unconstrained problem is not only sufficient but necessary in a sense.  Furthermore, in our present framework we will be able to write down  explicitly the optimal contract. We first derive the First Order Conditions (FOC) for Agent's and Principal's one-step problems:

\begin{lemma}
\label{FOCLagrange} 
Assume that $g^p_t$ is once and $g^a_t,c_t$ are twice continuously differentiable, for $t \in \T$. Then: 
%
\begin{equation}
(A,\gamma)\in \C_t(\beta) \quad \mbox{if and only if} \quad 
 \beta\mu-\nabla c_t(A)+\beta\sigma \nabla g^a_t(\gamma)=0. \label{eqCg}
\end{equation}
Moreover, given an optimal contract $\{(A_t,\beta_t,\gamma_t)\}$ for the Principal, and supposing for every time $t \in \T$ that the implied one-step contracts form a regular point for the corresponding constraints appearing in the r.h.s.\ of \eqref{eqCg} -this is, the matrices
$\Bigl [ \mu+\sigma \nabla g^a_t(\gamma_t) \mbox{ }\vert\mbox{ } \beta_t\sigma \nabla^2g^a_t(\gamma_t) \mbox{ }\vert\mbox{ }  -\nabla^2 c_t(A_t) \Bigr]\in \RR^{N\times (1+D+N)}$ have full range- there exists Lagrange multipliers $\lambda_t \in \RR^N$ s.t. the following systems admit a solution: 
\begin{align}
0&=\left[\beta_t\mu - \nabla c_t(A_t)\right] + \beta_t \sigma \nabla g^a_t(\gamma_t) \label{uno}\\ 
0&=[\mu-\nabla c_t(A_t)] + \sigma \nabla g^p(\sigma 'A_t-\gamma_t)-\nabla^2 c_t(A_t)\lambda_t  \label{dos}\\
0&=\nabla g^a_t(\gamma_t) - \nabla g^p_t(\sigma 'A_t-\gamma_t) + \beta_t \nabla^2 g^a_t(\gamma_t)\sigma'\lambda_t\label{tres}\\
0&=\lambda_t  [\mu +\sigma \nabla g^a_t(\gamma_t)].\label{cuatro}
\end{align}
\end{lemma}

\begin{proof}
We omit the time index for simplicity. The identity \eqref{eqCg} follows by differentiation and noticing that the optimization problem in $\C_t(\beta)$ is concave in the $A$ variable. It is also easy to see that the matrix $$\Big[ \mu+\sigma \nabla g^a(\gamma) \mbox{ }\vert\mbox{ } \beta\sigma \nabla^2g^a(\gamma) \mbox{ }\vert\mbox{ }  -\nabla^2 c(A) \Bigr ]\in \RR^{N\times (N+d+1)}$$ has as rows the gradients of the components of $\beta\mu-\nabla c_t(A)+\beta\sigma \nabla g^a_t(\gamma)$. 
By e.g.\ \cite[Chapter 3]{bstour} we thus have the existence of a Lagrange multiplier $\lambda$. Forming the Lagrangian 
$$L = \left[A\mu - c_t(A)\right]  +  g_t^a\left(\gamma \right)  + g_t^p\left(A\cdot\sigma-\gamma \right) +
\lambda\cdot\left\{ \left[\beta\mu - \nabla c_t(A)\right] + \beta \sigma \nabla g^a\left(\gamma \right)\right\}$$
and taking the partial derivatives w.r.t.\ $\lambda,A,\gamma,\beta$ yields the desired system.
\end{proof}

Dropping the time index again, notice that multiplying \eqref{uno} by $\lambda$ yields $\lambda\nabla c(A)=0$. Thus multiplying \eqref{tres} by $\lambda'\sigma$, \eqref{dos} by $\lambda$, adding them up and then multiplying by $\beta$ yields:
$$\beta \lambda' [\beta\sigma \nabla^2 g^a(\gamma)\sigma'-\nabla^2c(A)]\lambda=0.$$
Therefore, as soon as one searches for a $\beta > 0$ and either $c$ or $g^a$ are respectively strictly convex or concave, then necessarily $\lambda =0$. This shows that a reasonable optimal solution to the problem must necessarily solve also the ``unconstrained'' problem with FOC:
\begin{align*}
0&=\left[\beta\mu - \nabla c_t(A)\right] + \beta \sigma \nabla g^a(\gamma) \\ 
0&=[\mu-\nabla c_t(A)] + \sigma \nabla g^p(\sigma'A-\gamma) \\
0&=\nabla g^a(\gamma) - \nabla g^p(\sigma'A-\gamma).
\end{align*}
%
We knew from previous sections, in greater generality, that solving the unconstrained problem is sufficient to construct a solution to the original constrained one. Hence these last equations show that, in the present context at least, passing through the unconstrained formulation is actually also {necessary}, at least for contracts with $\beta>0$.

Subtracting the second from the first equation above and then using the third one, we get:
$$(\beta-1)[\mu+\sigma \nabla g^a(\gamma)]=0.$$  
Thus either $\beta=1$ is optimal, or $\mu+\sigma \nabla g^a(\gamma)= 0$. This last case can be called degenerate, since under it we derive from \eqref{uno} that it is optimal for the Agent to exercise minimum effort: $\nabla c(A)=0$. Since necessary conditions give a larger set of potential optimal points than the actual set of optima, we are inclined to say that this degenerate case is suboptimal. 

\subsection{Base preferences}

We close this section with an analysis of the benchmark case where both parties' preferences originate from a common base preference functional: $U^l(\cdot)= \frac{1}{\gamma^l}U(\gamma^l\cdot)$ for $l=a,p$. In terms of generators, this means that $g^l(\cdot)=\frac{1}{\gamma^l} g(\gamma^l \cdot)$. We assume that $\nabla g$ is injective. Then $(A^*,\gamma^*,\beta^{*})$ satisfies the system in Lemma \ref{FOCLagrange}, with $\lambda=0$, where $A^*$ solves
\begin{equation}
\label{markFOC}
0=[\mu-\nabla c_t(A^*)] + \sigma \nabla g\left[\frac{\gamma^a\gamma^p}{\gamma^a+\gamma^p}\sigma'A^*\right],
\end{equation}
and $\gamma^* = \frac{\gamma^p}{\gamma^a+\gamma^p}\sigma'A^*$ and $\beta^{*}=1$. 
%

{The final part of the following proposition shows to what extend the structure of optimal contracts in \cite[Theorem 1]{Ou-Yang} can be recovered. }

\begin{proposition}
\label{summmark}
Under the Markovianity and PRP Assumptions, the optimal contract (interpreted as a mapping between strategies to payments) is of the form of:
$$A \mapsto \bar{S}(A)= \kappa + \sum \gamma_t^*\Delta w_{t+1}+[W^A_T-\tilde{W}_T],$$
where $W^A_T=W_0+\sum  A_t\Delta \tilde{P}_{t+1} $, $\tilde{W}=W^{A^*}$, and $\kappa\in\RR$. Here $A^*$ and $\gamma^*$ (both vector/scalar valued deterministic processes) are the optimal ones for the Principal. Moreover, if the utilities stem from a common base functional, then we can write the optimal contract as:
$$ A \mapsto \bar{S}(A)= \bar{\kappa} + \frac{\gamma^p}{\gamma^p+\gamma^a}W^A_T+  \frac{\gamma^a}{\gamma^p+\gamma^a}[W^A_T-\tilde{W}_T] ,$$
{having the form of cash plus a convex combination of the wealth generated by the Agent and the performance (gains/losses) obtained w.r.t.\ a benchmark portfolio, as in \cite[Theorem 1]{Ou-Yang}.}

\end{proposition}

\begin{proof}
By Theorem \ref{master} we get:
\begin{equation*}
\Theta = R+\sum \left[ \gamma^*_t\Delta w_{t+1} + c_t(A_t^*)-A_t^*\Delta \tilde{P}_{t+1} - g_t^a(\gamma^*_t)  \right]= \kappa + \sum  \gamma^*_t\Delta w_{t+1} - \tilde{W}_T,
\end{equation*}
where we used that $\gamma^*$ and $A^*$ are optimal (Lemma \ref{FOCLagrange}), that $\beta=1$ is optimal, that $\kappa := R+\sum c(A_t^*)- g_t^a(Z^a_t+\sigma 'A_t^*)$ is a constant, thanks to Assumption \ref{suposmarkov}, and the fact that $A^*_t$ and $\gamma^*_t$ are deterministic (Remark \ref{non-randomn}). Again by Theorem \ref{master} this shows that the contract $A\mapsto \kappa + \sum  \gamma^*_t\Delta w + W^A_T-\tilde{W}_T$ is optimal. If further the utility functionals are a re-scaling of one another, we know that $\gamma_t^*=\frac{\gamma^p}{\gamma^p+\gamma^a}\sigma'A^*_t$. Plugging in this into the previous expression for the optimal contract, we conclude.
\end{proof}
%
%
%


%

\begin{example}[1d-Bernoulli Setting, Entropic Utility]
\label{ex1dentropia}
Suppose Agent's and Principal's utility functions are respectively
$$U^a_t(X)=-\frac{1}{\gamma^a}\log\left(\E\left[e^{-\gamma^a X}\vert \F_t \right] \right) \mbox{ and }U^p_t(X)=-\frac{1}{\gamma^p}\log\left(\E\left[e^{-\gamma^p X}\vert \F_t \right] \right),$$ 
with $\gamma^a,\gamma^p>0$, and that Agent's cost function is $c(a)=h\frac{a^2}{2}$. Assume also a one dimensional market driven by a simple Bernoulli-walk setting (that is $N=d=1$: one asset, one source of randomness); see example \ref{introbernoulli}. 
We first observe that $g_t(x)= -\log\left ( \frac{e^{\sqrt{h}x}+e^{-\sqrt{h}x}}{2} \right )= -log\circ \cosh(\sqrt{h}x)$, from which $\nabla g_t(x)= - \sqrt{h} \tanh(\sqrt{h}x)$. From here, and manipulating \eqref{markFOC}, we get that the optimal action $A_t^*$ at time $t$ is the solution to the equation:
$$-\frac{\gamma^a\gamma^p}{\gamma^a+\gamma^p} \sqrt{h} \sigma A_t^* \hspace{1pt}=\hspace{1pt} \frac{1}{2}\log\left ( \frac{\sigma\sqrt{h} + A_t^*h-\mu}{\sigma\sqrt{h} - A_t^*h+\mu}  \right ).$$  

\end{example}

\section{Conclusion}

The present article clarifies the structure of optimal linear contracts in dynamic models of portfolio delegation when both parties' preferences satisfy translation invariance, time consistency and certain regularity conditions. We have shown how the problem of dynamic contracting can be reduced to a recursive sequence of one-period conditional optimization problems. Using conditional analysis techniques we established general attainability results for the Agent and Principal problems and derived the representation of optimal contracts found in \cite{Ou-Yang} under a Markov-PRP assumption and for base preferences and general costs. Several questions are still open. First, the restriction to linear contracts is undesirable. Unfortunately, our method does in no obvious way carry over to non-linear contracts. Second, in the PRP framework we assumed that the Principal observes the driving process $\bar w$. Although this assumption seems common in the literature, it would be more natural to assume that the Principal observes the price increments only. This would add an additional adverse selection component to our model, if one interprets the Agent's additional information as his type, and hence leading to very complex optimization problems. Finally, it would be interesting to analyze portfolio delegation models under limited liability. If one restricts oneself a-priori to a particular class of pay-off profiles such as call options, then our methods can probably still be used to establish existence of optimal contracts (within the pre-specified class). It is an open questions how to analyze models of limited liability without any such a-priori restriction. 


\begin{appendix}

\section{Conditional Analysis}

This appendix recalls conditional analysis results needed to analyze our dynamic contracting problem. We also establish new results which are key to our PA problem. For a detailed discussion of finite dimensional conditional analysis we refer to \cite{CondRd} and references therein;  for a thorough treatment of conditional analysis on $L^p$ spaces we refer to \cite{FKV}. 

\subsection{Finite dimensional conditional analysis}

On a given probability space $(\Omega,\F,\Prob)$ we denote by $L$ and $L^0$ the sets of all, respectively all a.s.\ finite random variables. We apply almost-sure identification and ordering on this sets and put $\overline{L}:=\{X\in L:X>-\infty\}$ and $\underline{L}:=\{X\in L:X<\infty\}$ and denote by $\NN(\F)$ the set of variables in $L^0$ which take values in $\NN$. We fix $N\in \NN$ and view $E:=[L^0(\F)]^N$ as a finite-dimensional topological $L^0(\F)$-module over the ring $L^0(\F)$. On $E$ we define the \textit{conditional norm} $\| X \| = (X X)^{\frac{1}{2}}$ (notice that this is a random variable), where the product is the euclidean one.

\begin{definition}
\label{directorio}
A set $C \subset E$ is called:

\begin{itemize}

\item stable if $\ind{A}X+\ind{A^c}Y\in C$, for every $X,Y\in C$, $A\in\F$

\item $\sigma-$stable if $\sum_{n\in \NN}\ind{A_n}X_n \in C$, for every sequence $(X_n)\subset C$ and partition $(A_n)\subset\F$ of $\Omega$

\item $L^0-$convex if $\lambda X + (1-\lambda)Y \in C$, for every $X,Y\in C$ and $\lambda \in L^0$ with values in $[0,1]$

\item sequentially closed if it contains all the limits of its a.s.\ converging sequences.

\item $L^0-$bounded if $\esssup_{X\in C}\| X \| \in L^0 $.

\end{itemize}
\end{definition}

A stable and sequentially closed set is $\sigma-$stable. We define for $M\in \NN(\F)$ and $(X_n)\subset E$ the element $X_M = \sum_{n\in\NN}\ind{M=n}X_n \in E$ and notice that if the former sequence belongs to a $\sigma-$stable set, then the latter does so too. The following result is a generalization of the classical Bolzano-Weierstrass Theorem.

\begin{lemma}
\label{BW}
Let $(X_n)\subset E$ be $L^0-$bounded. Then there exists $X\in E$ and a sequence $(N_n)\in\NN(\F)$ such that $N_{n+1}>N_n$ and $X=\lim_{n\rightarrow\infty} X_{N_n}$ a.s.\, Also, if  $(x_n)\subset L^0$ is such that $x:=\limsup x_n \in L^0$, then there exists a sequence $(N_n)\in\NN(\F)$ such that $N_{n+1}>N_n$ and $x=\lim_{n\rightarrow\infty} x_{N_n}$ a.s.\
\end{lemma}
\begin{proof}
For the first statement refer to \cite[Theorem 3.8]{CondRd}. For the second, define $N_0=0$ and $N_{n}=\min\{m>N_{n-1}:x_m \geq x-1/n\}$. Then $N_n \in \NN(\F)$ and $N_{n+1}>N_n$, from which $N_n\geq n$ follows.
Now, notice that $\sup_{m\geq n} x_m \geq \sup_{m\geq N_n} x_m \geq x_{N_n} \geq x-1/n$ a.s., from which $x=\lim_{n\rightarrow\infty} x_{N_n}$ a.s.\
\end{proof}


As in the euclidean case, convexity opens the way to a necessary and sufficient characterization of boundedness (see \cite[Theorem 3.13]{CondRd}):

\begin{theorem}
\label{L0convbound}
Let $C$ be a sequentially closed $L^0-$convex subset of $E$ which contains $0$. Then $C$ is $L^0-$bounded if and only if for any $X\in C\backslash \{0\}$ there exists a $k\in \NN$ such that $kX \notin C$.
\end{theorem}

Let us now introduce the notions of continuity, convexity and stability of functions defined on subsets of $E$ and taking values in a set of random variables.  

\begin{definition}
\label{directorio2}
Let $C \subset E$. A function $f:C\rightarrow L$ is called:

\begin{itemize}

\item $L^0-$lower semicontinuous at $X\in C$ if $f(X)\leq \liminf f(X_n)$ for every sequence $(X_n)\subset C$ with a.s.\ limit $X$.

\item$L^0-$continuous at $X\in C$ if $f(X)= \lim f(X_n)$ whenever $(X_n)\subset C$ has a.s.\ limit $X$.

\item $L^0-$convex if $f(\lambda X + (1-\lambda)Y) \leq \lambda f(X) + (1-\lambda)f(Y)$, for every $X,Y\in C$ and $\lambda \in L^0$ with values in $[0,1]$

\item stable if $f(\ind{A}X+\ind{A^c}Y) = \ind{A}f(X)+\ind{A^c}f(Y)$, for every $X,Y\in C$, $A\in\F$.

\end{itemize}

\end{definition}

For the last two items it is assumed that $C$ is $L^0-$convex, respectively, stable. Strict $L^0-$convexity is defined in terms of a strict inequality. Finally $f$ is called (upper/lower semi)continuous on $C$ if it is so on every point of $C$. If $f$ is continuous and stable over a $\sigma-$stable and sequentially closed set, then it satisfies the stability property for countable partitions too. If $f$ is $L^0-$convex or $L^0-$concave, then it is local (meaning $\ind{A}f(X)=\ind{A}f(Y)$ whenever $\ind{A}X=\ind{A}Y$), which in itself directly implies that it also satisfies the stability property for countable partitions.

%

The following result is implied by the proof of \cite[Theorem 4.13]{CondRd}, since all the authors really use is $\sigma-$stability of the set under consideration (which is implied by their stronger assumptions). We give a self-contained proof here.

\begin{lemma}
\label{lemnoacotado}
If a non-empty set $C\subset E$ is $\sigma-$stable and is not $L^0-$bounded, then there is a set $\tilde{\Omega}$ with $\Prob(\tilde{\Omega})>0$ and a sequence $\{X_n\}\subset C$ such that, for every $n\in\NN$, $\vert X_n\vert\geq n$ over $\tilde{\Omega}$ 
\end{lemma}

\begin{proof}
We define $U_n:=\{B\in \F:\exists X\in C,\vert X\vert\geq n \mbox{ on }B\}$, which is non-empty since $C$ is unbounded, introduce the family of decreasing sets $A_n:=\left\{\esssup\limits_{B\in U_n} \ind{B} = 1 \right\}$ and put $A:=\bigcap_{n}A_n$. Assuming that $\Prob(A)=0$, or equivalently that $\Prob\left(\cup_n A_n^c\right)=1$, then for every $X\in C$:
$$\vert X\vert=\left\vert \sum_n X \ind{\left\{A_n^c \cap A_{n-1}\right\}} \right\vert \leq \sum_n \vert X\vert \ind{\left\{A_n^c \cap A_{n-1}\right\}} \leq \sum_n n \ind{\left\{A_n^c \cap A_{n-1}\right\}}\in L^0(\F) .$$
Since $X \in C$ was arbitrary, this implies that $C$ is $L^0(\F)-$bounded. Therefore $\Prob(A)>0$ must hold. By definition of $\esssup$ we have that there exist $\{B^{l,n}\}_l\in U_n$ such that $\esssup\limits_{B\in U_n}\ind{B}=\sup_l\ind{B^{l,n}}$ a.s.\ This implies $A_n=\bigcup_l B^{l,n}$ a.s.\ Taking $X^{l,n}$ such that $\vert X^{l,n}\vert\geq n$  on $B^{l,n}$, and fixing an $X^*\in C$ arbitrary, let us define:
$$X^{(n)}:= X^* \ind{\left\{(\bigcup_l B^{l,n})^c\right\}}+\sum_l X^{l,n}\ind{\left\{B^{l,n}\cap (\cup_{m<l}B^{m,n})^c\right\}}+ X^{0,n}\ind{B^{0,n}},$$
which belongs to $C$ thanks	to $\sigma-$stability. Clearly $$\vert X^{(n)}\vert\geq n\ind{\left\{\bigcup_l B^{l,n}\right\}} + \vert X^*\vert \ind{\left\{(\bigcup_l B^{l,n})^c\right\}},$$
and therefore a.s.\ $\vert\ind{A}X^{(n)}\vert \geq n\ind{A}$. Thus we have that $\vert X^{(n)}\vert\geq n$ on $A$ for every $n$. 
\end{proof}
    
The following conditional optimization theorem is used to prove attainability of the Agent Problem. For a proof we refer to  \cite[Theorem 4.4]{CondRd}.

\begin{theorem}
\label{CondOpt}
Let $C$ be a sequentially closed and stable subset of $E$ and $f:C\rightarrow \overline{L}$ be a $L^0-$lower semicontinuous, stable function. Assume there exists an $X_0\in C$ such that the set $\{X\in C:f(X)\leq f(X_0)\}$ is $L^0-$bounded.
Then there exists an $\hat{X}\in C$ such that 
$$f\left(\hat{X}\right)= \essinf\limits_{X\in C} f(X).$$
If $f$ and $C$ are $L^0-$convex then the ``$\argmin$'' set is also $L^0-$convex, and if in addition $f$ is strictly $L^0-$convex then $\hat{X}$ is the sole (a.s.) optimizer.

\end{theorem}


{We finally adapt a Komlos-type lemma (as in \cite[Lemma A1.1]{DelSch}) for conditionally bounded random variables, which we use to prove our general attainability result (Theorem \ref{teogeralasim}). We thank a referee for hinting at the proof we give now. }

\begin{lemma}
\label{teoDSenchulado}
 Let $\{\xi_n\}_n$ be $[0,+\infty)$-valued random variables defined on a common probability space $(\Omega,\G,\Prob)$, take $\F $ a sub-sigma algebra and assume that the set $C:=conv\{\xi_n : n\in \NN\}$ satisfies the following conditional boundedness condition:
\[ 
	\forall \epsilon \in L^0_+(\F),\exists a \in L^0(\F)\mbox{ such that }\forall h\in C, \Prob(h \geq a\vert \F)\leq \epsilon.
\]
Then there exists a $[0,+\infty)$-valued random variable $X$ and a sequence $\{x_n\}$, where $x_n$ belongs to the convex hull of $\{\xi_n,\xi_{n+1},\dots\}$ such that $x_n\rightarrow X$ almost surely.
 \end{lemma}

\begin{proof}
{By \cite[Lemma A1.1]{DelSch} it suffices to show that $C$ is bounded in probability. By assumption, we have that $  p_n:=\esssup_{h\in C} \Prob(h \geq n\vert \F) \to 0$,
as $n\to\infty$ and $\Prob-$a.s. Since also $p_n\in [0,1]$ a.s.\ we conclude by dominated convergence that $\E\left [ p_n \right ]\to 0$, which of course is stronger than 
$\sup_{h\in C}\Prob[h\geq n]\to 0$, so we conclude. }
\end{proof}

%

\subsection{Conditional analysis on $L^p$} 
Let $\F$ be a sub sigma-algebra of $\G$. For every $p\in [1,+\infty]$ we define:
\begin{eqnarray*}
\vert\vert X\vert\vert_p &=& \left\{
\begin{array}{cc}
\E[\vert X\vert^p\vert \F] & \hspace{10pt}\mbox{ if } p\in[1,+\infty)\\
\essinf \{Y\in L^0_+(\F) \mbox{ s.t. }Y\geq \vert X\vert\} & \mbox{if } p=+\infty.
\end{array}\right.
\end{eqnarray*}
This is well defined for every $X\in L^0(\G)$. We further define the conditional $L^p$-space $$L^p_{\F}(\G):=\{X\in L^0(\G)\mbox{ st. }\vert\vert X\vert\vert_p \in L^0(\F)\}.$$ It is shown in \cite{FKV} that $L^p_{\F}(\G)$ is a topological $L^0(\F)-$module over the topological ring $L^0(\F)$, and $\vert\vert \cdot\vert\vert_p$ is an $L^0(\F)-$norm inducing the module topology on $L^p_{\F}(\G)$. 

A function $U:L^p_{\F}(\G)\rightarrow \underline{L}^0$ is called:

\begin{itemize}

\item $L^0(\F)-$concave: if $U(\lambda X+(1-\lambda)X')\geq \lambda U(X)+(1-\lambda)U(X')$ for every $\lambda\in L^0(\F)\cap[0,1]$ and every $X,X'\in L^p_{\F}(\G)$

\item proper: if $\exists X\in L^p_{\F}(\G)$ such that $U(X)>-\infty$ and $\forall X'\in L^p_{\F}(\G)$ it holds $U(X)<\infty$

\item $L^p_{\F}(\G)$-upper semicontinuous: if for every net $\{X_{\alpha}\}\subset L^p_{\F}(\G)$ converging to some $X$ in conditional norm, it holds that $\essinf_{\beta}\esssup_{\alpha \geq \beta} U(X_{\alpha})\leq U(X)$

\item monotone: if $U(X)\geq U(X')$ whenever $X\geq X'$ 

\item translation invariant: if $U(X+Y)=U(X)+Y$ for every $X\in L^p_{\F}(\G)$ and $Y \in L^0(\F)$

\end{itemize}

The following representation result re-phrases \cite[Corollary 3.14]{FKV}:

\begin{theorem}
\label{Urep}
Let $p\in [1,\infty)$ and $U:L^p_{\F}(\G)\rightarrow \underline{L}^0(\F)$ satisfy the above conditions. Let $q$ be the H\"older conjugate of $p$ and define 
\[
	\mathcal{W}:=\{Z\in L^q_{\F}(\G):Z\geq 0, \E[Z\vert\F]=1\}, \quad 
	\alpha(Z):=\esssup_{X\in L^p_{\F}(\G)}\{ U(X)-\E[ZX\vert\F]\}.
\]
Then $U$ satisfies the following variational representation:
\begin{equation}
U(X)= \essinf\limits_{Z\in \mathcal{W}}\{ \E[ZX\vert\F] + \alpha(Z)\}. \notag
\end{equation}

\end{theorem}

In the next Proposition we prove that $L^p_{\F}(\G)-$upper semicontinuity is a consequence of $L^0-L^p$ upper semicontinuity (see Definition \ref{ranx}). This of course implies Proposition \ref{vorgeschmack}.

\begin{proposition}
\label{implicanciausc}
Let $U:L^p_{\F}(\G)\rightarrow \underline{L}^0(\F)$ be $L^0-L^p$ upper semicontinuous. Then $U$ is also $L^p_{\F}(\G)-$upper semicontinuous. Furthermore, if $U$ is also proper, monotone, translation invariant and $L^0(\F)$-concave, then $U$ admits a variational representation and for any $N\in\N$ and $\Delta \in [L^p_{\F}(\G)]^N$ the functional
\begin{equation}
A\in [L^0(\F)]^N \mapsto U(A \Delta) \label{bla}
\end{equation}
is $L^0$-upper semicontinuous in the sense of Definition \ref{directorio2}.{Under the same hypotheses, if $A_n\in [L^0(\F)]^N\to A$ a.s.\ and $\{\Gamma^n\}_n$ is $L^p_{\F}(\G)$-bounded such that  $\Gamma^n\to \Gamma$ a.s.\ then
$$\limsup_n U(A_n\Delta+\Gamma_n)\leq U(A\Delta +\Gamma)$$
}
\end{proposition}
\begin{proof}
For the first part, by \cite[Lemma 3.10]{FKV2}, it is enough to prove that the sets $K_k:=\{ X \in L^p_{\F}(\G): U(X)\geq k \}$ are conditionally closed for every $k\in L^0(\F)$. We will prove that their complements are conditionally open. To this end we fix such a $k$ and and assume to the contrary that $(K_k)^c$ is not open. We thus take $X$ such that $U(X)<k$ on a non-negligible set and such that for every $N\in \NN(\F)$ we have that $K_k\cap B(X,1/N)\neq \emptyset$, where $B(X,1/N)=\{Z:\E(\vert Z-X\vert^p\vert\F)\leq 1/N \}$. This means that we can find, for every $N\in \NN(\F)$, an element $X_N\in B(X,1/N)$ such that $U(X_N)\geq k$ a.s.\ A straightforward adaptation of Markov's inequality yields
\[
	\Prob(\vert X_N-X\vert\geq \epsilon \vert \F)\leq \frac{\E(\vert X_N-X\vert^p\vert\F)}{\epsilon^p}
\]
for every $\epsilon\in L^0(\F)_{++}$. From this we may find for every natural number $n$ an element $M_n\in \NN(\F)$ such that:

\begin{itemize}

\item for every $N\in\NN(\F)$ st. $N\geq M_n$ it holds that $\Prob(\vert X_N-X\vert\geq 1/n \vert \F)\leq 1/n^2$ a.s.\

\item for every n: $M_{n+1}> M_n$ a.s.\

\end{itemize}

Now, we will use a ``Borel-Cantelli Lemma''-type reasoning in order to prove that the sequence $\{X_{M_n}\}$ converges almost surely to $X$. First notice that for a fixed $l\in \NN$:
$$\sum\limits_{n\in\NN} \Prob(\vert X_{M_n}-X\vert\geq 1/l \vert \F)\leq \sum\limits_{n\leq l} \Prob(\vert X_{M_n}-X\vert\geq 1/l \vert \F) + \sum\limits_{n > l} \Prob(\vert X_{M_n}-X\vert\geq 1/n \vert \F),$$ 
and since the last term is bounded above by $\sum_{n>l}1/n^2$, the original sum belongs to $L^0(\F)$ (and so is a.s.\ finite). Define now $i.o.\left\{\vert X_{M_{\cdot}}-X\vert\geq 1/l \right\} := \bigcap_{m\in\NN} \bigcup_{n\geq m} \left\{\vert X_{M_n}-X\vert\geq 1/l \right\}$. Then:
$$\Prob \left( i.o.\left\{\vert X_{M_{\cdot}}-X\vert\geq 1/l \right\}\vert\F\right ) \leq \Prob\left(\bigcup_{n\geq m}\{ \vert X_{M_n}-X\vert\geq 1/l  \}  \vert\F\right)
\leq \sum_{n\geq m }  \Prob( \vert X_{M_n}-X\vert\geq 1/l   \vert\F),$$
and so the left-hand side does not depend on $m$ whereas the right one tends a.s.\ to $0$ as $m$ increases. This shows that $\Prob \left( i.o.\left\{\vert X_{M_{\cdot}}-X\vert\geq 1/l \right\}\vert\F\right )=0$ a.s. Taking expectations, $\Prob \left( i.o.\left\{\vert X_{M_{\cdot}}-X\vert\geq 1/l \right\}\right )=0$. Since this holds for every $l$, we conclude that indeed $\{X_{M_n}\}$ converges almost surely to $X$.

Finally we have by the $L^0-L^p$ upper semicontinuity assumption that $k\leq \limsup_n U(X_{M_n})\leq U(X)$ a.s.\ since by definition $X_{M_n}\in B(X,1)$, but also $U(X)< k$ on a non-negligible set, which is a contradiction. This completes the proof of the first statement.

By Theorem \ref{Urep} and the first claim we know that $U$ has a variational representation. {That $A\mapsto U(A\Delta)$ is $L^0$-upper semicontinuous is a consequence of the last claim in the proposition (taking $\Gamma_n=0$). So to establish the last claim and finish the proof, is suffices to compute 
$$\E\left[|A_n\Delta+\Gamma_n|^p|\F \right ]^{1/p} \leq C\sup_{i=1,\dots,N} |A_n^i|\E\left[|\Delta|^p|\F \right ]^{1/p}+\E\left[|\Gamma_n|^p|\F \right ]^{1/p}, $$
and observe that the r.h.s.\ is bounded from above by some r.v.\ in $L^0(\F)$, by conditional $L^p$-boundedness of the $\Gamma_n$ and since the (components of) the $A_n$ converge a.s. All in all  $\{A_n\Delta+\Gamma_n\}_n$ is $L^p_{\F}(\G)$-bounded and converges a.s.\ to $A\Delta+\Gamma$, so we conclude by the $L^0-L^p$ upper semicontinuity assumption.}
%
\end{proof}

\section{Optimized certainty equivalents and Proof of Theorem \ref{casodina}}

{We start with a number of technical results for the examples in Section \ref{preferences}. }

{
\begin{lemma}
\label{ejemx2}
{The following hold. 
\begin{itemize}
\item[(i)]
{The extensions in (\ref{eqextension}) are well defined for TVAR and, more generally, for  
optimized certainty equivalent families for which both $1\in \cap_t int(dom(H^*_t))$ and every $H_t$ is bounded from below (equivalently $0\in dom(H_t^*)$).}
\item[(ii)] 
Take $\F=\F_t$, $\G=\F_{t+1}$ for $t$ fixed. For any $\gamma>0$ and $\lambda\in (0,1)$, the entropic functional 
$$X\in L^1_{{\F}}(\G) \mapsto -\frac{1}{\gamma}\log\left(\E(exp(-\gamma X)|\F)\right),$$
as well as the Tail-value-at-risk functional 
$$X\in L^1_{{\F}}(\G)\mapsto \esssup_{s}\left\{s-\lambda^{-1}\E([s-X]_+|\F)  \right\},$$
are $L^0-L^1$ u.s.c. More generally, optimized certainty equivalents for which $1\in int(dom(H^*))$ are $L^0-L^1$ u.s.c. 
\item[(iii)]The TVAR family and, more generally, OCE families for which $1\in \cap_t int(dom(H^*_t))$ and every $H_t$ is bounded from below, all satisfy Assumption \ref{suposUasymm} after pasting. 
\end{itemize}}
%

\end{lemma}
}
\begin{proof}
{For TVAR we have $H(l)=\lambda^{-1}[l]_+$ and $H^*(x)=\Psi_{[0,\lambda^{-1}]}(x)$, the convex indicator of $[0,\lambda^{-1}]$. In particular, $1\in int(dom(H^*))$ and $H$ is lower bounded. 
%
In the following we work with abstract OCEs for which the latter conditions hold.}

\begin{itemize}
{
	\item[i)] It suffices to show that the extensions defined by \eqref{eqextension} produce a functional which never attains the value $+\infty$. Let $U$ stand for any OCE (associated to $H$) satisfying the stated properties and let $X$ be a r.v.\ not attaining $+\infty$. For $1+\epsilon\in int(dom(H^*)) $ we define
\begin{align*}
N^{\epsilon} &:= \,\, \sum_{n=1}^{\infty}n \ind{\Prob(X\leq n |\F)>[1+\epsilon]^{-1},\, \Prob(X\leq n-1 |\F)\leq [1+\epsilon]^{-1}} \, =\,\, \inf\{n\in\N:\Prob(X\leq n |\F)>[1+\epsilon]^{-1}\},
\end{align*}
which is then finite and belongs to $\N(\F)$. Inspired by \cite[Proof of (12)]{Equilibrium} we introduce a partition $A_0:=\{\Prob(X\leq N^{\epsilon} |\F)>0\}$ and $A_n:=\{\Prob(X\leq N^{\epsilon}+n |\F)>0,\, \Prob(X\leq N^{\epsilon}+ n-1 |\F)=0\}$ for $n\geq 1$, so we define
$$ \xi=\sum_{n\geq 0} \ind{A_n}\frac{\ind{X\leq N^{\epsilon}+n}}{\Prob(X\leq N^{\epsilon}+n|\F)}.$$
It is then easy to see that $\E[\xi|\F]=1$ and that $\xi \in [0,1+\epsilon]$, so that a.s.\ $\xi\in dom(H^*)$. We conclude that
\begin{eqnarray*}
	U(X\wedge k) &=& \esssup_s \{s-\E[H(s-X\wedge k)]\} \\
	& \leq & \esssup_s \{s-s\E[\xi|\F]+\E[\xi(X\wedge k)+H^*(\xi)|\F]\}\\
	& \leq & \E[\xi X|\F]+\E[H^*(\xi)|\F] \leq \sum_{n\geq 0}\ind{A_n}[N^{\epsilon}+n]+\E[H^*(\xi)|\F]<+\infty, \notag
\end{eqnarray*} 
where finiteness comes from the fact that $H^*$ must send $[0,1+\epsilon]$ into a bounded set. Hence  $\lim_{k\to +\infty}U(X\wedge k)<+\infty$ and we conclude.}


\item[ii)] Let us take $X_n$ bounded in $L^1_{\F}(\G)$ such that $X_n\rightarrow X$ a.s.\ For any $C\in L^0(\F)$ we want to show that if $\esssup_{s}\left\{s-\E(H(s-X_n)|\F)  \right\}\geq C$ then also $\esssup_{s}\left\{s-\E(H(s-X)|\F)  \right\}\geq C$. Indeed, let us first take $s_n\in L^0(\F)$ such that
$$s_n-\E(H(s_n-X_n)|\F)\geq C-n^{-1}. $$
%
%
%
Because $H$ is convex, lower-semicontinuous and proper, we have that $H(s_n-\E(X_n|\F))\geq R[s_n-\E(X_n|\F)]-H^*(R)$ for each $R$, and so in particular
$$s_n[1-R]\geq C-n^{-1}-H^*(R)-R\E(X_n|\F),$$
for every $R\in dom(H^*)$. If we were able to find $R\in dom(H^*)\cap (1,\infty)$ then for such element we would have 
$$s_n\leq \frac{C-n^{-1}-H^*(R)-R\E(X_n|\F)}{1-R}.$$
Similarly, if $r\in dom(H^*)\cap (-\infty,1)$ existed then we would get 
$$s_n\geq \frac{C-n^{-1}-H^*(r)-r\E(X_n|\F)}{1-r}.$$
Altogether, we could conclude that the quantities $s_n$ are $L^0(\F)$-bounded, since the random variables $X_n$ were bounded in $L^1_{F}(\G)$. By Lemma \ref{BW} we could find $N_n\in\N(\F)$ increasing to $+\infty$ such that $s_{N_n}\rightarrow \bar{s}$ a.s.\, for some $\bar{s}\in L^0(\F)$, and obviously $X_{N_n}\rightarrow X$ a.s.\ still. By locality we would have that 
$$s_{N_n}-\E(H(s_{N_n}-X_{N_n})|\F)\geq C-{N_n}^{-1}. $$
Taking $\limsup_n$, using the fact that $H$ is bounded below by an affine function and conditional Fatou's Lemma, we may obtain
$$\bar{s}-\E(H(\bar{s}-X)|\F)\geq C. $$
This readily implies what we wanted to prove. To conclude, observe that the conditions $dom(H^*)\cap (1,\infty)\neq\emptyset$ and $dom(H^*)\cap (-\infty,1)\neq\emptyset$ are together equivalent to $1\in int(dom(H^*))$ in our case, since $1\in dom(H^*)$ by definition and $H^*$ is convex. 

{\item[iii)] Upper semincontinuity has already been dealt with. Under the stated conditions pasting is also already justified, so it only remains to show the condition on the domain of $U_t$. If $X\in dom(U_t) $, then there must be some $s\in\R$ such that $\E[H(s-X)|\F_t]<\infty$. But as $H$ is bounded from below, this implies $\E[H(s+X_-)|\F_t]<\infty$, and by the normalization property on $H$ (implying $H(l)\geq l$) we get in turn $\E[X_-|\F_t]<\infty$ as desired.
}
\end{itemize}
\end{proof}

{
\begin{remark}\label{remtechnical}
The reason we had to prove \cite[Proof of (12)]{Equilibrium} anew in the first part of the preceding proof, is that we want to include the case where $dom(H^*)\neq [0,\infty)$, in order to cover e.g.\ the TVAR family. This creates the difficulty of finding a $\xi$ satisfying simultaneously $\E[\xi|\F]=1$, $\xi \in [0,1+\epsilon]$ and $\E[\xi X|\F]<\infty$. 
\end{remark}
}

\begin{remark}
\label{oceobs}
In case $U^a_t$ arises as an optimized certainty equivalent (see example \ref{oce}), we know by Lemma \ref{ejemx2} and Proposition \ref{implicanciausc} that it has a variational representation. Notice then by Young inequality that
$$\alpha_t(Z)\leq \esssup_X\esssup_s\{s-\E[H_t(s-X)|\F_t]-\E[ZX|\F]\}\leq \E[H_t^*(Z)|\F_t],$$
whenever $Z\in L^{\infty}_{\F_t}(\F_{t+1})$ is s.t.\ $\E[Z|\F_t]=1$. In \cite{Equilibrium} and \cite{Divutil} it is proved that under given conditions there is equality above (for $Z\in\mathcal{W}_t$). In any case we see that the conditions in Theorem \ref{teogeralasim} on $K^{a,p}$ are satisfied if these utility functionals are such that $1\in int(dom(H_t^*))$; indeed, we may just take $\epsilon>0$ such that $[1-\epsilon,1+\epsilon]\subset dom(H^*_t) $.
\end{remark}

Equipped with the previous result, we provide the proof of our general existence of optimal contracts. 
 
\hspace{0.5mm} 
 
{\sl Proof of Theorem \ref{casodina}.} We may assume $\F=\F^A$. Under the given condition on the $H$'s, the condition on the $K$'s (see Theorem \ref{teogeralasim}) is satisfied, thanks to Remark \ref{oceobs}. By Proposition \ref{Genmetateo} it remains to show that $h_{t+1}\in dom(U^p_t)$ for each $t$.

Either using Young's inequality or invoking Remark \ref{oceobs}, we know that $U_t^a(X),U_t^p(X)\leq \E[X|\F_t]$. From this we see
\begin{align*}
V_t(A,\Gamma)&\leq -c(A)
+A\E[\Delta \tilde{P} \vert\F_t]+\E[h_{t+1}\vert \F_t]\\
&\leq -\left(\frac{c_t(A)}{|A|}-\frac{2p_+}{p_-}\right)\vert A\vert + \E[h_{t+1}\vert \F_t]\\ &\leq K_t + \E[h_{t+1}\vert \F_t],
\end{align*}
where we used that $\vert \Delta \tilde{P}_{t+1} \vert=\vert diag(P_t)^{-1}\Delta P_{t+1} \vert \leq \frac{2p_+}{p_-}$ and the existence of $K_t\in\RR$ is a consequence of the growth of $c_t$ and its continuity. From here we get by definition that $h_{t} \leq   K_t + \E[h_{t+1}\vert \F_t]$ and since $h_T=0$ by backwards inductions follows that $h^p_{t+1}\leq L$ for some constant $L$ and all $t$. Monotonicity of the $U^p$'s mean that $h_{t+1}\in dom(U^p_t)$.
\hfill $\Box$ 
%
%
\section{Proof of Theorem \ref{master}}

First we turn our attention to the Agent's recursion. Let $\bar{a}$ be a generic sequence of efforts. From equation \eqref{eqlarga}, we see that defining 
$$H_t(\bar{a}_t,\dots,\bar{a}_{T-1}) := U_t^a\left( \Theta(P_{0:T}) +  \sum_{s\geq t}\left\{ \beta_s\Delta W^{\bar{a}}_{s+1} -c_{s+1}(\bar{a}_{s+1})\right\} \right) - c_t(\bar{a}_t) , $$
we get the recursion 
\begin{align*}
H_T (\bar{a}_t,\dots,\bar{a}_{T-1})&= \Theta(P_{0:T}) ,\\
H_t(\bar{a}_t,\dots,\bar{a}_{T-1}) &=  U_t^a\left( H_{t+1}(\bar{a}_{t+1},\dots,\bar{a}_{T-1}) + \beta_t\bar{a}_t \Delta\tilde{P}_{t+1} \right)  - c_t(\bar{a}_t)  .
\end{align*}
Then, in terms of $H_t:=\esssup_{a_{t},\dots,a_{T-1}} H_t(a_{t},\dots,a_{T-1})$, we get: 
$$H_t(\bar{a}_{t},\dots,\bar{a}_{T-1})\leq -c(\bar{a}_t)+ U^a_t\left(\esssup_{a_{t+1},\dots,a_{T-1}} H_{t+1}(a_{t+1},\dots,a_{T-1}) +\beta_t \bar{a}_{t}\Delta \tilde{P}_{t+1}\right).$$
This yields that $H_t\leq \esssup_{a_t}\left\{-c(a_t)+ U^a_t\left(H_{t+1} + \beta_t a_{t}\Delta \tilde{P}_{t+1}\right)\right\}$.
For $t=T-1$ this is an equality and by assumption the value $H_{T-1}$ is attained at some $\hat{a}_{T-1}$. Suppose now that equality holds in the previous equation for $t+1,\dots,T-1$, and $H_{t+1}$ was attained say at $(\hat{a}_{t+1},\dots,\hat{a}_{T-1})$. This implies that:
\begin{align*}
H_t & \leq \esssup_{a_t}\left\{-c(a_t)+ U^a_t\left(H_{t+1}(\hat{a}_{t+1},\dots,\hat{a}_{T-1}) + \beta_t  a_{t}\Delta \tilde{P}_{t+1}\right)\right\}\\
 &\leq \esssup_{a_t,\dots,a_{T-1}}\left\{-c(a_t)+ U^a_t\left(H_{t+1}(a_{t+1},\dots,a_{T-1}) + \beta_t a_{t}\Delta \tilde{P}_{t+1}\right)\right\}\\
 &= \esssup_{A_t,\dots,A_{T-1}} H_{t}(A_{t},\dots,A_{T-1}) \hspace{3pt}=:\hspace{3pt} H_t.
 \end{align*}
So indeed at time $t$ also $H_t = \esssup_{a_t}\left\{-c(a_t)+ U^a_t\left(H_{t+1} + \beta_t  a_{t}\Delta \tilde{P}_{t+1}\right)\right\}$ holds and by assumption the last term is attained at some $\hat{a}_t$,
from which $H_t$ is attained at $(\hat{a}_{t},\dots,\hat{a}_{T-1})$. This closes the inductive step, and therefore the desired recursion holds.

Now we will establish rigorously recursion \eqref{intermhG} (equivalently \eqref{intermh}). To this end we denote by $\beta=(\beta_t)_t$ a generic decision variable for the Principal and $a=(a_t)_t$ where $ a_t\in L^0(\F_t)^N$, a corresponding optimal effort for the Agent. Let 
$$N:=\sum_{s\geq t+1} \left[(1-\beta_s)a_s\Delta\tilde{P}_{s+1} -\Delta H_{s+1} \right].$$ 
Then using the just proven expression for $H_t$ (i.e.\ \eqref{recHbas}), and setting $\Gamma=\beta_t a_t\Delta\tilde{P}_{t+1} + H_{t+1}$, we get:
\small
\begin{align*}
U_t^p\left( \sum_{s\geq t} \left[(1-\beta_s)a_s\Delta\tilde{P}_{s+1} -\Delta H_{s+1} \right] \right) =& U_t^p\left( (1-\beta_t)a_t\Delta\tilde{P}_{t+1} - H_{t+1}  - c_t(a_t) + N\right .\\\vspace{1pt}
&+ \left . U_t^a\left( H_{t+1}+ \beta_t a_t\Delta\tilde{P}_{t+1} \right)\right),\\
=& U_t^p\left( a_t\Delta\tilde{P}_{t+1} - \Gamma +  U_t^a\left( \Gamma  \right) - c_t(a_t) + N\right),\\
=& U_t^a\left( \Gamma  \right) - c_t(a_t) + U_t^p\left( a_t\Delta\tilde{P}_{t+1}-\Gamma + N \right ).
\end{align*}
\normalsize
And now, applying time-consistency and translation invariance in the last term above we get:\small
$$ U_t^p\left( \sum_{s\geq t} \left[(1-\beta_s)a_s\Delta\tilde{P}_{s+1} -\Delta H_{s+1} \right] \right) = U_t^a\left( \Gamma  \right) - c_t(a_t) + U_t^p\left( a_t\Delta\tilde{P}_{t+1}-\Gamma + U_{t+1}^p(N) \right ).$$ \normalsize
Therefore calling $h_{t+1}(a,\Gamma)=U_{t+1}^p( N )$, we obtain recursion \eqref{intermhG}. That is to say, if $(a,\Gamma)$ does not satisfy this recursion, they will not be chosen by the Principal. In the same way we conclude for $(a,\beta)$ and recursion \eqref{intermh}. With these recursions for $h_t(\cdot)$ already established, we can proceed to prove \eqref{rechas} the same way we proved the recursion for $H$. First recall that actually $h_t(a,\Gamma)$ is a short-hand for $h_t((a_s,\Gamma_{s+1})_{s\geq t})$. From this and \eqref{intermhG} we have: 
$$ h_t((\bar{a}_{s},\bar{\Gamma}_{s+1})_{s\geq t})\leq U_t^p\left(\esssup_{A,\Gamma} h_{t+1}(A,\Gamma) +\bar{a}_{t}\Delta \tilde{P}_{t+1} - \bar{\Gamma}_{t+1}   \right) + U_t^a(\bar{\Gamma}_{t+1})-c_t(\bar{a}_t)   .$$
This yields 
\begin{align*}
h_t&
\leq  \esssup\limits_{(a_{t},\Gamma_{t+1})\in \C_t(\beta_t) }  U_t^p\left(h_{t+1} +a_{t}\Delta \tilde{P}_{t+1} - \Gamma_{t+1}   \right) + U_t^a(\Gamma_{t+1}) -c_t(a_t)  .
\end{align*}
For $t=T-1$ this is an equality (we defined $h_T=0$) and by assumption the value $h_{T-1}$ is attained. Using induction, similarly as how we did it for $H$, we get that \eqref{rechas} holds.

The validity of the change of variables $\beta_ta_t\Delta\tilde{P}_{t+1}+H_{t+1} \rightarrow \Gamma_{t+1}$ and the introduction of $\C(\beta)$ as a constraint inducing {incentive compatibility} are now obvious. This means that $h$ represents the future wealth prospects of the Principal. Hence at time $t=0$ we obtain a solution for the whole Principal's problem, proving as well that Principal's optimal wealth is $W_0-R+h_0$. 

We proceed now to prove that a solution to Principal's recursion delivers indeed an optimal (dynamic) contract, and that the Agent behaves as predicted. Call $(\beta_t,A_t,\Gamma_{t+1})_{t}$ the optimal quantities attaining $h$ in \eqref{rechas}. Define $\Theta$ and the contract $S$ as in the statement of the present theorem. Then: 
\small
\begin{align*}
 U_{T-1}^a\left( \Theta + \beta_{T-1}a_{T-1} \Delta\tilde{P}_{T} \right)  - c_{T-1} (a_{T-1} ) =&  \sum_{0\leq t< T-1} \left[ \Gamma_{t+1} - \beta_{t}A_{t}\Delta\tilde{P}_{t+1} - U^a_t(\Gamma_{t+1}) + c_t(A_t)  \right] \\
 &R+[c_{T-1}(A_{T-1} )  - c_{T-1} (a_{T-1} ) ] - U^a_{T-1}(\Gamma_T)  \\
 &+U_{T-1}^a\left( \Gamma_T - \beta_{T-1}A_{T-1} \Delta\tilde{P}_{T} + \beta_{T-1}a_{T-1} \Delta\tilde{P}_{T} \right).
\end{align*}
\normalsize
By definition of $\C(\beta)$ the sum of the last terms is smaller or equal than $0$, and exactly zero when $a_{T-1}=A_{T-1}$. Therefore
\small
\begin{align*}
\esssup_{a_{T-1}} \left\{U_{T-1}^a\left( \Theta + \beta_{T-1}a_{T-1} \Delta\tilde{P}_{T} \right)  - c_{T-1} (a_{T-1} )\right\} =&  R + \sum_{0\leq t< T-1} [ \Gamma_{t+1} - \beta_{t}A_{t}\Delta\tilde{P}_{t+1} + c_t(A_t) - U^a_t(\Gamma_{t+1})  ].
\end{align*}
\normalsize
This shows that at time $T-1$ the Agent chooses $A_{T-1}$ when presented with $(\Theta,\beta)$. If we define $H_T=\Theta$, we are thus entitled to call $H_T$ the value (left hand side or right one) in the above equality. By using backwards induction, as we have often done and hence omit, 
we have proven that the contract $S$ (defined from $(\Theta,\beta)$) is optimal for the Principal and incentive compatible (notice that automatically $H_0=R$), and the Agent indeed chooses $A$ under this contract. \hfill $\Box$

\end{appendix}

\bibliographystyle{plain}
\bibliography{mybib}

\end{document}